\documentclass[11pt,a4paper]{article}

\usepackage[utf8]{inputenc}
\usepackage{amsthm,amsmath,latexsym,amssymb,amsfonts,amscd,cite}
\usepackage{amssymb} 
\usepackage[all]{xy}
\usepackage{tikz}
\usepackage{array}
\usepackage{hyperref}

\usepackage{tocbasic}
\DeclareTOCStyleEntry[
  beforeskip=.2em plus 1pt,
  pagenumberformat=\textbf
]{tocline}{section}

\usepackage{mathrsfs}
\usepackage{hyperref}

\usepackage{bbm}
\usepackage{bm}
\usepackage[T1]{fontenc}
\usepackage{mathrsfs}

\numberwithin{equation}{section}

\theoremstyle{definition}

\newtheorem{proposition}[equation]{Proposition}
\newtheorem{remark}[equation]{Remark}
\newtheorem{definition}[equation]{Definition}
\newtheorem{example}[equation]{Example}

\usepackage[all]{xy}
\usepackage{graphicx}

\usepackage{mathtools}

\makeatletter
\DeclareRobustCommand\widecheck[1]{{\mathpalette\@widecheck{#1}}}
\def\@widecheck#1#2{%
    \setbox\z@\hbox{\m@th$#1#2$}%
    \setbox\tw@\hbox{\m@th$#1%
       \widehat{%
          \vrule\@width\z@\@height\ht\z@
          \vrule\@height\z@\@width\wd\z@}$}%
    \dp\tw@-\ht\z@
    \@tempdima\ht\z@ \advance\@tempdima2\ht\tw@ \divide\@tempdima\thr@@
    \setbox\tw@\hbox{%
       \raise\@tempdima\hbox{\scalebox{1}[-1]{\lower\@tempdima\box
\tw@}}}%
    {\ooalign{\box\tw@ \cr \box\z@}}}
\makeatother

\newcommand{\sfLambda}{\mathsf{\Lambda}}
\newcommand{\sfGamma}{\mathsf{\Gamma}}
\newcommand{\sfZ}{\mathsf{Z}}
\newcommand{\sfH}{\mathsf{H}}

\newcommand{\ee}{\mathrm{e}}
\newcommand{\dd}{\mathrm{d}} 

\newcommand{\sfs}{\mathsf{s}} 
\newcommand{\sft}{\mathsf{t}} 
\newcommand{\sfi}{\mathsf{i}} 
\newcommand{\sfm}{\mathsf{m}} 
\newcommand{\sfp}{\mathsf{p}} 

\title{{\bf Symplectic Groupoids \\ and Poisson Electrodynamics}}

\date{}
\begin{document}

\maketitle

\begin{center}
\vskip -0.05\textheight
\renewcommand{\thefootnote}{\fnsymbol{footnote}}
Vladislav G. \textsc{Kupriyanov}${}^{1,2,}$\footnote{{\tt vladislav.kupriyanov@gmail.com}},  Alexey A. \textsc{Sharapov}${}^{1,2,}$\footnote{{\tt sharapov@phys.tsu.ru}} \& Richard J. \textsc{Szabo}${}^{3,}$\footnote{{\tt R.J.Szabo@hw.ac.uk}}
\renewcommand{\thefootnote}{\arabic{footnote}}
\setcounter{footnote}{0}

\end{center}

\begin{center}
${}^1$
\emph{Centro de Matem\'atica, Computa\c{c}\~ao e Cogni\c{c}\~o} \\
\emph{Universidade
Federal do ABC, Santo Andr\'e, SP, Brazil}\\[3mm]
${}^2$
\emph{Physics Faculty, Tomsk State University}\\
\emph{Lenin ave. 36, Tomsk
634050, Russia}\\[3mm]
${}^3$
\emph{Department of Mathematics, Heriot--Watt University}\\ \emph{Colin Maclaurin Building,
            Riccarton, Edinburgh EH14 4AS, U.K.}\\ and \emph{Maxwell Institute for
            Mathematical Sciences, Edinburgh, U.K.}
\end{center}

\vspace{5mm}

\begin{abstract}
\noindent
    We develop a geometric approach to Poisson electrodynamics, that is, the semi-classical limit of noncommutative $U(1)$ gauge theory. Our framework is based on an integrating symplectic groupoid for the underlying Poisson brackets, which we interpret as the classical phase space of a point particle on noncommutative spacetime. In this picture gauge fields arise as bisections of the symplectic groupoid while gauge transformations are parameterized by Lagrangian bisections. We provide a geometric construction of a gauge invariant action functional which minimally couples a dynamical charged particle to a background electromagnetic field. Our constructions are elucidated by several explicit examples, demonstrating the appearances of curved and even compact momentum spaces, the interplay between gauge transformations and spacetime diffeomorphisms, as well as emergent gravity phenomena.
\end{abstract}

\begin{flushright}
		\small
		{\sf EMPG--23--11}
	\end{flushright}
 
{\baselineskip=12pt
\tableofcontents
}

\section{Introduction}
\label{sec:intro}

The idea of curved momentum space dates back to the early days of quantum field theory \cite{Born}. It stems from the desire to effectively bound the momenta $p$ of particles from above and thereby overcome problems with ultraviolet divergences.    For example, if the space of momenta happens to be compact, then the desired upper bound is determined by the volume of the momentum space. This idea has resurged from time to time in different forms \cite{born1949reciprocity,Golfand:1962kjf,Tamm,Kadyshevsky:1977mu,KOWALSKI_GLIKMAN_2013,Franchino-Vinas:2023rcc,Amelino} but has met with limited success so far. 

As noted by Born in \cite{born1949reciprocity}, the hypothesis of curved momentum space is closely related to the idea of noncommutative spacetime, where the positions $x$ of particles cannot be measured with arbitrary accuracy \cite{MBron,Snyder1,Snyder2}.  The intuitive reasoning is that it is impossible to bring two repulsing particles arbitrarily close if their momenta are bounded from above. Consequently, it becomes impossible to probe spacetime at arbitrarily small distances and speak, for example, about a point where two colliding particles interact. This brings one to the idea of noncommutative spacetime with its `smeared' notion of a point. The proposal to get rid of local interactions, and hence of ultraviolet divergences, through spacetime noncommutativity was first spelt out by Snyder in  \cite{Snyder1}.

In \cite{Born, born1949reciprocity} the interplay between the curvature of momentum space and spacetime noncommutativity was formulated as a {\it reciprocity principle}. This naive principle requires all physical laws to be invariant under the canonical transformations $x\mapsto -\alpha\, p, \, p\mapsto \alpha^{-1}\, x$ for any non-zero constant $\alpha$.  Then the curvature of a spacetime manifold in general relativity implies the curvature of momentum space.\footnote{For a recent discussion of the reciprocity principle, see \cite{Buoninfante_2021}.} In this paper we revise and develop a more modern version of Born's reciprocity principle within a precise mathematical framework. In our approach, the reciprocity principle is built into  the structure of {\it symplectic  groupoids}. 

The formal definition, basic properties and some examples of symplectic groupoids are discussed in \S\ref{sec:symplgr}. For now, let us give a preliminary idea of how they work. Consider linear commutation relations for the spacetime coordinates given by
    \begin{equation}\label{xxx}
        [x^i, x^j]=f^{ij}_k\, x^k \ .    
        \end{equation}
The Jacobi identity for the commutator forces $f^{ij}_k$ to be the structure constants of some Lie algebra $\mathfrak{g}$.  The dual space $X=\mathfrak{g}^\ast$ of $\mathfrak{g}$  thus plays the role of a noncommutative spacetime. 
Let $G$ be a Lie group which integrates the Lie algebra $\mathfrak{g}$.  
    
Regarding the commutator (\ref{xxx}) as coming from the quantization of a linear Poisson bracket on $X$, one can integrate the latter to a symplectic  groupoid.  This is given by the cotangent bundle $T^\ast G$ with its canonical symplectic structure. Then the symplectic manifold  $T^\ast G\simeq X \times G$ is identified with the phase space of a point particle living in the noncommutative spacetime $X$. This identification automatically interprets $G$ as the space of particle momenta. 
If we further suppose that the Killing metric on $\mathfrak{g}$ is positive-definite, then $G$ is compact as a topological space. The compactness of $G$ 
imposes an upper bound on the uncertainty of momenta given by $\Delta p \leqslant \mathrm {Vol}(G)^{1/n}$, where $\mathrm {Vol}(G)$ is the volume of $G$ in its Haar measure and $n=\dim(G)$. The Heisenberg uncertainty relation\footnote{As we will see in \S\ref{sec:symplgr},  the positions and momenta have non-zero Poisson brackets on $T^\ast G$.} then sets a lower bound on the uncertainty of positions given by $\Delta x \geqslant \frac \hbar 2\,\mathrm {Vol}(G)^{-1/n}$. This can be regarded as a manifestation of the general reciprocity principle: the more noncommutative a spacetime is, the more curved the space of conjugate momenta becomes. 

A field theory realisation of curved momentum space was obtained in \cite{Freidel_2006} from the effective dynamics of scalar fields coupled to three-dimensional quantum gravity. After integrating out the metric field in the path integral,  one is left with an effective field theory on a noncommutative spacetime with the commutation relations (\ref{xxx}) of the Lie algebra  $\mathfrak{su}(2)$. The Feynman diagram expansion for the effective quantum field theory shows that the momenta of virtual particles take their values in the Lie group $SU(2)$ rather than in a three-dimensional vector space. 
In~\cite{KOWALSKI_GLIKMAN_2013} the same conclusion about momentum space is arrived at by considering the coupling of a classical point particle to three-dimensional gravity, while in~\cite{Lust:2017bgx} it arises from dimensional reduction of M-waves probing four-dimensional M-theory backgrounds with non-geometric Kaluza--Klein monopoles. Thus non-linear phase spaces like $T^\ast G$ appear in some concrete dynamical models, for example through the interaction with topological modes of fields. In our treatment, a non-linear  phase space with curved momentum subspace is postulated from the very beginning.  

Gauge fields can be naturally incorporated into this picture by considering the interactions of open strings on D-branes with background $B$-fields. For a single D-brane wrapping $X=\mathbb{R}^n$ in a constant $B$-field, or equivalently with a constant Poisson structure, the  seminal treatment of Seiberg and Witten in \cite{Seiberg_1999} shows that the low-energy effective field theory on the D-brane has a consistent deformation to a noncommutative gauge theory. The $U(1)$ gauge potential $A$ has field strength
\begin{equation}\label{F}
F_{ij}=\partial_i A_j-\partial_j A_i+[A_i,A_j]_\star
\end{equation}
and transforms under infinitesimal $*$-gauge transformations as
\begin{equation}\label{SW2}
    \delta_\varepsilon A_i=\partial_i\,\varepsilon +[A_i,\varepsilon]_\star \quad , \quad \delta_\varepsilon F_{ij} = [F_{ij},\varepsilon]_\star \ ,
\end{equation}
for a $U(1)$ gauge parameter $\varepsilon,$ which close to the Lie algebra
\begin{align} \label{eq:starclose}
    [\delta_{\varepsilon_1},\delta_{\varepsilon_2}] = \delta_{[\varepsilon_1,\varepsilon_2]_\star} \ .
\end{align}
The brackets $[-,-]_\star$ denote star-commutators defined with respect to the Moyal--Weyl star-product of functions which quantizes the constant Poisson structure. See~\cite{Douglas:2001ba,Szabo:2001kg} for early reviews of the subject.

The noncommutative gauge theory admits a well-defined consistent low-energy limit for slowly varying gauge fields $A$. The low-energy limit appears as the semi-classical limit where one replaces  the star-commutators in (\ref{F})--\eqref{eq:starclose} with the corresponding Poisson brackets. The Poisson bracket brings non-linearity into the Maxwell equations for the $U(1)$ gauge potentials without introducing higher derivatives or non-localities. It also upgrades the abelian gauge transformations of $A$ to nonabelian gauge transformations. We refer to the resulting gauge theory as {\it Poisson electrodynamics}.

Non-constant background $B$-fields lead to non-constant Poisson brackets, so that the extent of spacetime noncommutativity varies from point to point. This also yields a more general Lie algebra of infinitesimal gauge symmetries. However, the  naive implementation of this generalisation encounters immediately the following problem: the partial derivatives $\partial_i$ in (\ref{F}) do not differentiate non-constant Poisson brackets by the Leibniz rule. As a result, neither the covariance of the field strength in \eqref{SW2} nor the closure of gauge variations in \eqref{eq:starclose} holds. The problem of constructing a consistent noncommutative field theory in this case is related to the mathematical problem of constructing a differential graded Poisson algebra on differential forms.

By now there are several proposals for how to incorporate a general Poisson bracket into noncommutative field theory. In this paper we are interested in the recent proposal of \cite{Kupriyanov:2018xji,Kupriyanov:2018yaj,Kupriyanov_2021} which advocates the method of {\it symplectic embedding}.\footnote{The terminology is used because the technique also covers the more general cases of twisted and almost Poisson structures. In this paper we work solely with Poisson structures, for which a symplectic embedding coincides with the more common notion of {\it symplectic realization} that is used in the mathematics literature, see \S\ref{sec:symplgr}.} As the  name suggests, one starts with extending a given Poisson manifold $X$, regarded as a physical spacetime, to some ambient symplectic manifold $G\subseteq T^\ast X$. The  latter is endowed with nondegenerate Poisson brackets of the form
\begin{equation}\label{se}
\{x^i,x^i\}=\pi^{ij}(x)\quad,\quad \{x^j, p_i\}=\gamma^j_i(x,p)\quad,\quad \{p_i, p_j\}=0\ .
\end{equation}
Here $\pi^{ij}(x)$ is the Poisson bivector on $X$, $p_i$ are linear coordinates in the fibers of the cotangent bundle $T^\ast X$, and $\gamma^i_j(x,p)=\delta^i_j+O(p)$ are given by formal power series in $p_i$ with smooth coefficients. The symplectic structure is thus defined in a formal neighbourhood of the zero section $X\subset T^\ast X$. The construction of the functions  $\gamma_i^j$ for the most part repeats the construction of a local symplectic groupoid  integrating the Poisson manifold $(X, \pi)$, see \S\ref{sec:symplgr}. 

Gauge transformations are then  postulated  in  the form
\begin{equation}\label{gtr}
\delta_\varepsilon A_i=\gamma_i^j (x,A)\,\partial_j\varepsilon +\{A_i, \varepsilon\}\ ,
\end{equation}
where $\varepsilon(x)$ is an infinitesimal gauge parameter. 
Remarkably, these transformations close to the Lie algebra 
\begin{equation}\label{commut}
[\delta_{\varepsilon_1},\delta_{\varepsilon_2}]=\delta_{\{\varepsilon_1,\varepsilon_2\}} \ .
\end{equation}
One of the aims of the present paper is to give a conceptual geometric explanation  for the formula (\ref{gtr}). 
Again the language of symplectic groupoids is most suitable for this purpose, and their appearance in this context is in hindsight not surprising:
Local symplectic groupoids appear in the semi-classical limit of quantizations of Poisson manifolds~\cite{KM}.

\begin{table}
\begin{center}
\begin{tabular}{ |p{2mm}|p{7.3cm}|p{7.3cm}|} 
 \hline
 &&\\

1 &\bf Semi-classical spacetime & Poisson manifold $(X, \pi)$\\[3mm]
 
2&\bf  Phase space of a point particle in  $X$ & Symplectic groupoid $G\rightrightarrows X$ integrating $(X, \pi)$ \\ [7mm]
 
 3 &\bf Electromagnetic potentials  &  Group $\mathscr{B}(G)$ of bisections \\ [3mm]
  
 4&\bf Gauge group &  Subgroup $\mathscr{L}(G)\subset \mathscr{B}(G)$ of Lagrangian bisections \\ [7mm]

5& \bf Gauge transformations&  Left action of $\mathscr{L}(G)$ on $\mathscr{B}(G)$ accompanied by right action on $G$: If $\Sigma\mapsto \Sigma'\cdot\Sigma$ for $\Sigma\in \mathscr{B}(G)$ and $\Sigma'\in \mathscr{L}(G)$, then $$R^{-1}_{\Sigma'}: G\longrightarrow G$$\\[-5mm]

6& \bf Minimal coupling to electromagnetic field & Described by the Hamiltonian $$H^\Sigma=R^\ast_\Sigma H$$ where $H\in C^\infty(G)$ is the Hamiltonian of a
neutral  particle and $\Sigma \in \mathscr{B}(G)$ is a gauge field\\ [ 30mm]

7& \bf  Gauge invariant electromagnetic field strength tensor&    Two-form $F^\sft=\Sigma_\sft^\ast\omega$ on $X$, where $\Sigma_\sft\in \mathscr{B}(G)$ is a section of the target map and $\omega$ is the symplectic form on $G$\\[13mm]

8& \bf Gauge covariant  electromagnetic field strength tensor&  Two-form $F^\sfs=\Sigma_\sfs^\ast\omega$ on $X$, where $\Sigma_\sfs\in \mathscr{B}(G)$ is a section of the source map\\
&& \\
 
\hline
\end{tabular}
\end{center}
\caption{\small Physical quantities in Poisson electrodynamics (in the left column) and their geometric realisations in the language of symplectic groupoids (in the right column). Items 1--2 are explained in detail together with explicit examples in \S\ref{sec:symplgr}, items 3--5 in \S\ref{sec:bisections}, and items 6--8 in \S\ref{sec:CPE}.\label{tab:glossary}}
\end{table}

In our geometric setting, gauge fields are parameterized by the group of \emph{bisections} of a symplectic groupoid, whose subgroup of Lagrangian bisections parametrize (finite) gauge transformations of Poisson electrodynamics. This generalises the situation in classical electrodynamics, where a $U(1)$ gauge field $A$ on $X$ is a section of the cotangent bundle $T^*X$ and gauge transformations act as translations  of $A$  by closed one-forms. Generally, bisections act by diffeomorphisms on the base manifold $X$, thus realising the interplay between gauge transformations and spacetime diffeomorphisms that is a general feature of noncommutative gauge theories~\cite{Gross:2000ba,Lizzi:2001nd}. We give precise definitions, properties and examples in~\S\ref{sec:bisections}.

Recent work on Poisson electrodynamics is found in~\cite{Kupriyanov:2021aet,Kurkov:2021kxa,Abla:2022wfz,Kupriyanov:2022ohu,Kupriyanov:2023gjj,Abla:2023odq}. One puzzling feature of the theory is that it is not clear how to couple it to matter fields in representations of the gauge group other than the adjoint representation: gauge symmetries and interactions in noncommutative field theories constructed solely in terms of star-products (rather than star-commmutators) do not have a well-defined semi-classical limit. In \S\ref{sec:CPE} we take first steps towards addressing this problem by constructing minimal couplings of dynamical charged particles to a background electromagnetic field in Poisson electrodynamics. We further give natural definitions of both gauge covariant and gauge invariant field strength tensors for the gauge potential, and briefly address the problem of constructing a gauge invariant action functional for the electromagnetic field. Our constructions are purely geometric and are summarised in Table~\ref{tab:glossary}. 

We proceed to present several explicit examples of our constructions in \S\ref{sec:ex}, including the case of a constant Poisson structure discussed above where we show how emergent gravity is a feature of Poisson electrodynamics. We further exemplify the considerations of curved momentum space in the $SU(2)$ example mentioned above as well as in the popular example of $\kappa$-Minkowski spacetime. Finally, we illustrate our constructions for a class of non-linear Poisson brackets.

In \S\ref{sec:concl} we conclude with a summary of our results and some prospects for future work.

\newpage

\section{Symplectic groupoids}
\label{sec:symplgr}

We begin by introducing some basic concepts and examples whose systematic exposition can be found in \cite{KM, WdS, MC, Visman, CFM}. For a gentle introduction to the subject of Lie groupoids,  see \cite{weinstein1996groupoids}. 

\paragraph{\underline{{\textsf{Groupoids}.}}}

\begin{definition}
A  \emph{groupoid} is  a small category whose morphisms are all invertible. 
\end{definition}

To unravel this definition let us denote by $X$ the set of objects and by $G$ the set of morphisms of the category. Then there are two maps  $$\xymatrix{ & G \ar[dr]^\sft \ar[dl]_\sfs & \\ X & & X }$$ called {\it source} and {\it target}. It is convenient to denote morphisms $g\in G$ by arrows 
$$x \xrightarrow{ \ g \ } y$$ with $x=\sfs(g)\in X$ and $y=\sft(g)\in X$. 

Two arrows $g_1, g_2\in G$ are called {\it composable} if $\sft(g_1)=\sfs(g_2)$. By definition,  to each pair of composable arrows $(g_1, g_2)$ there corresponds an arrow $g=g_1\,g_2 \in G$ called the composition (or product) of  $g_1$ and $g_2$:
$$
\xymatrix{x \ar[r]^{g_1}\ar@/_1pc/[rr]_{g_1\,g_2}&  y\ar[r]^{g_2} & z}
$$
The product of composable arrows is required to be associative:
$$(g_1\,g_2)\,g_3=g_1\,(g_2\,g_3)\ .$$

For each $x\in X$ there exists an identity morphism $e_x\in G$, called a \emph{unit}, such that $\sfs(e_x)=\sft(e_x)=x$, $e_x\, g=g$ and $g\,e_y=g$ for every morphism
\smash{$x\stackrel{g}{\longrightarrow} y$}.  Furthermore, to each morphism \smash{$x\stackrel{g}{\longrightarrow} y$} there corresponds an inverse morphism \smash{$y\xrightarrow{ \ g^{-1} \ }x$} such that $g^{-1}\,g=e_x$ and $g\,g^{-1}=e_y$:
$$
\xymatrix{ x\ar@/^/[rr]^{g}\ar@(ul,dl)[]_{e_x}& & y\ar@/^/[ll]^{g^{-1}}\ar@(ur,dr)[]^{e_y}
}
$$

It is convenient to identify $X$ with the subset of units  and think of $G$ as a special bifibration $\sfs,\sft: G\rightarrow X$ over the base $X$, see Figure~\ref{Gr}. This suggests the shorthand notation $G\rightrightarrows X$
for a groupoid.

\begin{figure}[ht]
\center{
\Large
\begin{tikzpicture}[scale=1.5]
\draw[] (0,0) -- (5,0);
\draw[-, color=blue] (0.25,-1.5) -- (2,2);
\draw[-, color=blue] (1.25,-1.5) -- (3,2);
\draw[-, color=blue] (2.25,-1.5) -- (4,2);
\draw[-, color=blue] (3.25,-1.5) -- (5,2);
\draw[-, color=blue] (0.75,-1.5) -- (2.5,2);
\draw[-, color=blue] (1.75,-1.5) -- (3.5,2);
\draw[-, color=blue] (2.75,-1.5) -- (4.5,2);

\draw[-, color=red ] (0,2) -- (1.75,-1.5);
\draw[- , color=red] (1,2) -- (2.75,-1.5);
\draw[-, color=red] (2,2) -- (3.75,-1.5);
\draw[-, color=red] (3,2) -- (4.75,-1.5);
\draw[-, color=red] (0.5,2) -- (2.25,-1.5);
\draw[-, color=red] (1.5,2) -- (3.25,-1.5);
\draw[-, color=red] (2.5,2) -- (4.25,-1.5);

\draw[-, thick, color=blue] (2,1) -- (2.25,1.5);
\draw[-, thick, color=blue] (2.5,0) -- (2.75, 0.5);

\draw[-, thick, color=red] (2.5,0) -- (2,1);
\draw[-, thick, color=red] (2.75,0.5) -- (2.25, 1.5);

\fill [black]  (2.5,0) circle (1 pt);
\fill [black]  (2,1) circle (1 pt);
\fill [black]  (2.75,0.5) circle (1 pt);
\fill [black]  (2.25, 1.5) circle (1 pt);
\fill [black]  (2,-1) circle (1 pt);

\coordinate [label=right:${}_{{}_{h}}$] (A) at (2.75, 0.5);
\coordinate [label=left:${}_{{}_{g}}$] (A) at (2, 1);
\coordinate [label=right:${}_{{}_{\!\!g\,h}}$] (A) at (2.25, 1.5);
\coordinate [label=right:${}_{{}_{g^{{-\!1}{}}}}$] (A) at (1.95, -1);
\coordinate [label=below:${}_{{}_{\sft(\!g\!)=\sfs(\!h\!)}}$] (A) at (2.51, 0.08);
\coordinate [label=:${}_{X}$] (A) at (0.25, 0);
\coordinate [label=left:${\mbox{\tiny{\textcolor{blue}{$\sfs$-fibres}}}}$] (A) at (5.2, 1);
\coordinate [label=right:${\mbox{\tiny{\textcolor{red}{$\sft$-fibres}}}}$] (A) at (-0.2, 1);

\coordinate [label=:${}_{X}$] (A) at (0.25, 0);

\end{tikzpicture}
\normalsize
}
\caption{\small Groupoid multiplication.}
\label{Gr}
\end{figure}
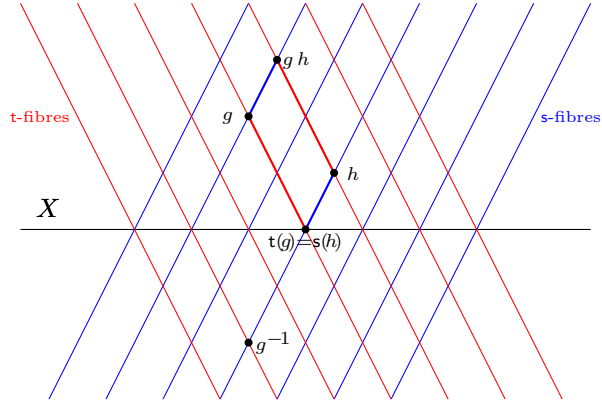

\begin{example}\label{GG}
Every group $G$ can be viewed as a groupoid over the one-point base $X=\{e\}$. Hence $\sfs(g)=\sft(g)=e\in G$ and all  morphisms $g\in G$ are composable.  The associative product is given by multiplication in the group $G$. This  example explains the origin of the name `groupoid'. 
\end{example}

Notice that the elements of the intersection $\sfs^{-1}(x)\cap \sft^{-1}(x)$, being pairwise composable, form a group $G_x$ for each $x\in X$. This group is called the {\it isotropy group} at $x\in X$; in Example~\ref{GG}, $G_e\simeq G$. The {\it orbit} of a point $x\in X$ is the set $O_x:=\sft(\sfs^{-1}(x))$ of all targets that can be reached from a given source $x$. A groupoid $G\rightrightarrows X$ is called {\it transitive}  if $O_x =X$ for some (and hence for all) $x\in X$; this means that for each pair of points $x,y\in X$ there exists at least one arrow \smash{$x\stackrel{g}{\longrightarrow}y$}. We say that a groupoid  $G\rightrightarrows X$ is {\it totally intransitive} if $O_x=\{x\}$ for all $x\in X$.

\paragraph{\underline{\textsf{Lie groupoids.}}}

In the following we consider groupoids endowed with additional structures. In particular, all of our groupoids will be {\it smooth}, meaning that $G$ and $X$ are smooth manifolds and all structure maps --- source, target, multiplication, and inversion --- are smooth.\footnote{Usually one also requires the source and target maps to be surjective submersions.} These are called {\it Lie groupoids}.

\begin{example}\label{Ex2}
Let $T G$ be the tangent bundle of a Lie group $G$. The group $G$ acts on itself by left and right translations:
$$
g\longmapsto h\,g\qquad \text{and} \qquad g\longmapsto g\, h\ ,
$$
for all $h\in G$.
Let $l(h)$ and $r(h)$ denote the corresponding maps of the tangent space $T_gG\rightarrow T_{h\,g}G$ and $T_gG\rightarrow T_{g\,h}G$. 

The elements of the tangent bundle $T G$ are pairs $(g, v)$, where $g\in G$ and $v\in T_gG$. One can make $TG$ into a Lie groupoid by identifying the submanifold of units $X\subset TG$ with the tangent space $T_e G$ at the identity $e\in G$.  The source and target maps are defined as
$$
\sfs(g,v)=\big(e, r(g^{-1})\,v\big) \qquad \text{and} \qquad \sft(g,v)=\big(e, l(g^{-1})\,v\big)\ .
$$
Hence the elements $(g,v)$ and $(h,u)$ are composable if and only if
\begin{equation}\label{comcon}
l(g^{-1})\,v=r(h^{-1})\,u \ \in \ T_eG\ .
\end{equation}
For composable pairs we set 
$$
(g,v)\,(h,u)=L_g(h,u):=\big(g\,h, l(g)\,u\big) \ \in \ T_{g\,h}G\ ,
$$
where $L_g: TG\rightarrow TG$ is  the diffeomorphism of the tangent bundle induced by the left translation on $G$.   The inversion is given by
$$
(g, v)^{-1} = \big(g^{-1}, l(g^{-1})\,r(g^{-1})\,v\big)\ .
$$
\end{example}

\paragraph{\underline{\textsf{Symplectic groupoids.}}}

Of most interest to us are Lie groupoids endowed with compatible symplectic structures. These are called `symplectic groupoids'. 
To give their formal definition, we need some more notation. 
Given a pair of symplectic manifolds  $(M_1,\omega_1)$ and $(M_2,\omega_2)$, one can endow the product manifold $M_1\times M_2$ with the symplectic structure  determined by the sum of  pullbacks
$$
\omega_1\oplus \omega_2:=\sfp_1^\ast\,\omega_1+\sfp^\ast_2\,\omega_2 \ ,
$$
where $\sfp_1$ and $\sfp_2$ are the canonical  projections onto the first and second factors in $M_1\times M_2$, respectively. Given a Lie groupoid $G\rightrightarrows X$, let $G\ast G\subseteq G\times G$ denote the submanifold of composable elements:
\begin{equation}
    G\ast G=\{ (g_1,g_2)\in G\times G\; | \; \sft(g_1)=\sfs(g_2)\}\ .
\end{equation}

\begin{definition}
    A \emph{symplectic groupoid} is a Lie groupoid $G\rightrightarrows X$ such that
    \begin{enumerate}
        \item[(a)]  $G$ is a symplectic manifold with symplectic two-form $\omega$, and
        \item[(b)]\label{item:b} the multiplication map 
        $$
        \sfm:G\ast G\longrightarrow G \ ,\quad (g,h)\longmapsto g\,h
        $$
        is Poisson, i.e., the pullback of the symplectic form $\omega$ by $\sfm$ coincides with the restriction of $\omega\oplus \omega$ to $G\ast G$. 
    \end{enumerate}
\end{definition}   

Generally, a $k$-form $\alpha\in \sfLambda^k(G)$ is called {\it multiplicative} if $$\sfm^\ast (\alpha) =\alpha \oplus\alpha \big|_{G\ast G}\ .$$
According to item (b), the symplectic structure on $G$ is defined by a multiplicative two-form $\omega$. It follows from the definition that
\begin{enumerate}
    \item[S1.] the inverse map ($g\mapsto g^{-1}$) is an antisymplectomorphism of $(G, \omega)$, i.e.,  it sends $\omega$ to $-\omega$;
    \item[S2.] $X$ is a Lagrangian submanifold of $(G,\omega)$;
    \item[S3.] $X$ has a unique Poisson structure $\pi$ for which $\sfs$ is a Poisson map and $\sft$ is anti-Poisson, i.e.,~$\sfs_*(\omega^{-1})=\pi$ and $\sft_*(\omega^{-1})=-\pi$; 
    \item[S4.] the symplectic leaves of $(X, \pi)$ coincide with the orbits of $G\rightrightarrows X$. 
\end{enumerate}

Property S3 allows one to refer to $(G, \omega)$ as a {\it symplectic realization} of the Poisson manifold $(X,\pi)$. The main idea is to `desingularize' a given Poisson structure $\pi$ by constructing a Poisson submersion from a suitable symplectic manifold $(G, \omega)$ onto $(X, \pi)$.  

\begin{example}\label{TG}
        We provide a symplectic realization for any linear Poisson bracket. In Example~\ref{Ex2}, one can replace $TG$ with the cotangent bundle $T^\ast G$ to produce a new Lie groupoid.\footnote{In fact, this is true for any tensor bundle over $G$.} The cotangent bundle $T^\ast G$ carries the canonical symplectic structure given by the differential of the Liouville one-form, which makes $T^\ast G$ into a symplectic groupoid. 
        
        Indeed, each fibre $T^\ast_gG$ is a Lagrangian submanifold, and in particular the submanifold of units $T^\ast_e G\subset T^\ast G$ is Lagrangian. To check property S3 it is convenient to trivialize the cotangent bundle $T^\ast G$ using the left translations:
        \begin{equation}\label{triv}
        T^\ast_g G \ \ni \ (g,v)\longrightarrow \big(g, l(g^{-1})\,v\big) \ \in \ G\times T^\ast_eG\ .
        \end{equation}
        Here we view the vector space $T^\ast_e G = \mathfrak{g}^\ast$ as the Lie coalgebra dual to the Lie algebra $\mathfrak{g}$ of the group $G$. The corresponding coproduct $\Delta: \mathfrak{g}^\ast\rightarrow \mathfrak{g}^\ast\otimes \mathfrak{g}^\ast$ then endows the space $\mathfrak{g}^\ast$ with a linear Poisson bracket. If $\{t_a\}$ is  a basis of $\mathfrak{g}$ such that $[t_a,t_b]=f_{ab}^c\,t_c$, then in the dual basis $\{t^a\}$ we get $\Delta t^a=f^a_{bc}\,t^b\wedge t^c$; here $\langle t^a,t_b\rangle=\delta^a_b$.
        
        Under the trivialization (\ref{triv}) the source and target maps become
        \begin{equation}\label{lpb}
            \sfs(g,v)= v \qquad \text{and} \qquad \sft(g,v)=\mathsf{Ad}_{g^{-1}}^\ast(v) \ ,
        \end{equation}
        for all $(g,v)\in G\times\mathfrak{g}^\ast$, where $\mathsf{Ad}^\ast$ stands for the coadjoint representation of $G$. Hence the product of composable elements is given by $(g,v)\,(h,u)=(g\,h, v)$. 
        The canonical symplectic structure on $T^\ast G$ then takes the form
        $$
        \omega=\dd\langle v, \dd g\, g^{-1}\rangle=\langle \dd v\wedge \dd g\, g^{-1}\rangle-\langle v,\dd g\, g^{-1}\wedge \dd g\, g^{-1}\rangle
        =\dd v_a \wedge e^a-v_a\, f^a_{bc}\,e^b\wedge e^c\ ,
        $$
where $v_a$ are linear coordinates on $\mathfrak{g}^\ast$ and $\dd g\, g^{-1}=e^a\,t_a$ is the right-invariant one-form on $G$ with values in  $\mathfrak{g}$. 

It is enough to check multiplicativity for the symplectic potential $$\theta=\langle v, \dd g\,g^{-1}\rangle \ .$$ One can identify the submanifold of composable elements with $G\times G\times \mathfrak{g}^\ast$ through the identifications
\begin{equation*}
    \big((g,v)\,,\, (h, \mathsf{Ad}^\ast_{g^{-1}}(v))\big)\quad\longleftrightarrow\quad (g,h,v)\ .
\end{equation*}
Under this identification 
\begin{equation*}
    \sfp_1(g,h,v)=(g,v)\quad,\quad \sfp_2(g,h,v)=\big(h, \mathsf{Ad}^\ast_{g^{-1}}(v)\big)\quad,\quad \sfm(g,h,v)=(g\,h,v)\ ,
\end{equation*}
giving
\begin{equation*}
    \sfp_1^\ast\, \theta=\langle v, \dd g\, g^{-1}\rangle \qquad \text{and} \qquad \sfp^\ast_2\, \theta=\langle \mathsf{Ad}^\ast_{g^{-1}} (v), \dd h\, h^{-1}\rangle\ .
    \end{equation*}
Hence
\begin{align*}
    \sfm^\ast \theta = \langle v, \dd(g\,h)\,(g\,h)^{-1}\rangle &= \langle v, \dd g\,g^{-1}\rangle + \langle v, g\, (\dd h\, h^{-1})\,g^{-1}\rangle \\[4pt]
    &= \langle v, \dd g\, g^{-1}\rangle + \langle \mathsf{Ad}^\ast_{g^{-1}}(v), \dd h \, h^{-1}\rangle  = \sfp_1^\ast\,\theta + \sfp_2^\ast\,\theta  \ ,     
    \end{align*}
and multiplicativity follows. 
\end{example}

\paragraph{\underline{\textsf{Local symplectic groupoids.}}}
 
A closer look at Example~\ref{TG} shows that the symplectic realization for linear Lie--Poisson structures owes its existence  to Lie's third theorem, that is, to the existence of a Lie group integrating a given Lie algebra. One can wonder about symplectic realizations for an arbitrary (not necessarily linear) Poisson structure: can any Poisson manifold $X$  be `integrated' to a symplectic groupoid having $X$ as the manifold of units? Unfortunately, the answer to this question is no, with many explicit counterexamples. 

However, every Poisson manifold can be integrated to a {\it{local symplectic groupoid}}. The latter can be roughly defined as a tubular neighbourhood of the base $X$ of a genuine symplectic groupoid (if it really existed). The proof is based on the following simple construction \cite[Ch.~II]{KM}. 

Let $(U, q^i)$ be a coordinate chart on an $n$-dimensional Poisson manifold $(X,\pi)$. Consider the system of ordinary differential equations
\begin{equation}\label{ODE}
\frac{\dd q^i}{\dd t}=\pi^{ij}(q)\, p_j\ , \quad i=1,\dots,n \ ,
\end{equation}
where $p$ is a fixed vector of $\mathbb{R}^n$. Let $q(t)=q(x,p,t)$ be the unique solution of (\ref{ODE}) subject to the initial condition $q(0)=x\in U$. For $p$ sufficiently small, we can assume that $q(t)$ exists for all $t\in [-1,1]$. Define functions
\begin{equation}\label{int}
\sigma^i(x, p)=\int^{1}_{0}\, q^i(x,p,t) \, \dd t\qquad \text{and} \qquad \tau^i(x,p)=\int_{-1}^0\, q^i(x,p,t)\, \dd t\ .
\end{equation}
Solving the equations 
\begin{equation}\label{equ}
q^i=\sigma^i(x,p)\qquad\text{and}\qquad q^i=\tau^i(x,p)
\end{equation}
for $x$ respectively gives two vector functions
\begin{equation}\label{sol}
x^i=\sfs^i(q,p) \qquad \text{and} \qquad x^i=\sft^i(q,p)\ .
\end{equation}
By construction $\sfs^i(q,0)=\sft^i(q,0)=q^i$.   

Let $\omega=\dd p_i\wedge \dd q^i$ be the canonical symplectic form on a sufficiently small  domain $W\subset \mathbb{R}^{2n}$ parametrized by $(q^i,p_j)$ on which the solutions (\ref{sol}) to the equations (\ref{equ}) make sense. Geometrically, one can view $W$  as a neighbourhood of the zero section in $T^\ast U$. Lemma 1.4 of \cite[Ch.~II]{KM} states  that
the functions (\ref{sol}) define, respectively, Poisson and anti-Poisson maps $W\to U$ that Poisson commute with each other:     
\begin{equation}\label{s-t}
     \{\sfs^i, \sfs^j\}=\pi^{ij}(\sfs) \quad,\quad    \{\sft^i, \sft^j\}=-\pi^{ij}(\sft) \quad,\quad   \{\sfs^i, \sft^j\}=0 \ ,
     \end{equation}
where $\{-,-\}$ denotes the Poisson bracket of $(X,\pi)$.

The vector functions (\ref{sol}) define the source and target maps of a local symplectic groupoid over $U$. It remains to describe a partially-defined multiplication operation on $W$. To this end,  one introduces the Hamiltonian flow $\varphi_t$ on $W$ generated by some function $\sfs^\ast h$ with $h\in C^\infty(U)$. For each point $a=(q,p)\in W$ it is always possible to  choose the Hamiltonian $h$ in such a way that $a=\varphi_1(\sft(a))$. Let $b=(q',p')$ be another point in $W$ such that $\sfs(b)=\sft(a)$. We put
\begin{equation}\label{ab}
     a\,b:=\varphi_1(b)\ .
     \end{equation}
According to Lemma 1.7 of \cite[Ch.~II]{KM}, this product satisfies all the required properties and does not depend on the choice of a Hamiltonian function $h$.  
By carefully gluing up the domains $W$ along $X$, one obtains a local symplectic groupoid integrating the Poisson structure $\pi$. 

\begin{remark}
As mentioned in~\cite{Kupriyanov_2021}, this construction enables one to make contact with the symplectic embedding method. For this, we note that the formulas (\ref{sol}) and (\ref{s-t}) imply the Poisson brackets
\begin{equation}\label{PBr}
    \{x^i,x^j\}=\pi^{ij}(x)\quad,\quad \{x^j, p_i\}=\frac{\partial\, \sfs^j}{\partial q^i}\quad,\qquad \{p_i,p_j\}=0\ .
\end{equation}
This allows one to identify the tensors $\gamma$ in (\ref{se}) as 
\begin{equation*}
    \gamma_i^j(x,p)=\left(\frac{\partial\, \sfs^j}{\partial q^i}\right)\big(\sigma(x,p)\,,\,p\big)\ .
\end{equation*}
\end{remark}

One can use the formulas (\ref{ODE})--(\ref{sol}) and (\ref{ab}) to explicitly integrate constant Poisson structures. 

\begin{example}\label{zero-PB} 
Every manifold $X$ can be viewed as a Poisson manifold   with the zero Poisson bracket. The corresponding symplectic groupoid 
integrating the zero Poisson structure is just the total space of the cotangent bundle $T^\ast X$ endowed  with the canonical symplectic structure $\omega=\dd p_i\,\wedge\, \dd x^i$. The source and target maps coincide with each other as well as with the canonical projection $\sfp: T^\ast X\rightarrow X$, while the object inclusion $X\to T^*X$ is the zero section. Then $\sfs^{-1}(x)=\sft^{-1}(x)=T_x^\ast X$ and the groupoid $T^\ast X$ is totally intransitive. As a result, the isotropy group at each point $x\in M$ is maximal and is given by $G_x=T_x^\ast X$; here  we regard the fibre $T_x^\ast M$ as an abelian group for the addition of covectors: $(p,p')\mapsto p+p'$. The product of composable elements reads
$$
(x,p)\,(x,p')=(x, p+p')\ .
$$
The simple identity 
$$
\dd x^i\wedge \dd(p_i+p'_i)=\dd x^i\wedge \dd p_i +\dd x^i\wedge \dd p'_i
$$
shows that the canonical  symplectic structure $\omega$ is multiplicative. 
\end{example}

\begin{example}\label{const-PB}
    A slightly more non-trivial example is provided by the symplectic realization of constant Poisson brackets. Consider the vector space $X=\mathbb{R}^n$ equipped with a constant bivector $\pi$. The corresponding symplectic groupoid is now given by the symplectic vector space $\mathbb{R}^{2n}=\mathbb{R}^n\times \mathbb{R}^n$ endowed with the canonical symplectic  form $\omega=\dd p_i\wedge \dd q^i$ with $i=1,\ldots, n$. 
     In coordinates, the source and target maps are defined by the Bopp shifts
     $$
     \sfs^i(q,p)=q^i-\tfrac12\,\pi^{ij}\,p_j\qquad\text{and}\qquad \sft^i(q,p)=q^i+\tfrac12\,\pi^{ij}\,p_j\ .
     $$
     The orbit $\sft\circ \sfs^{-1}(q)$ is given by the affine symplectic plane in $(\mathbb{R}^n, \pi)$ passing through the point $(q,0)$,  so that the groupoid is transitive if and only if the bivector $\pi$ is nondegenerate.  
     
    The product of two composable elements $(q,p)$ and $(q',p')$ with 
    $$
    \sft(q,p)=x=\sfs(q',p') \ \in \ \mathbb{R}^n
    $$
    reads
    \begin{equation}\label{comp}
    \begin{array}{c}
    \big(x^i-\frac12\,\pi^{ij}\,p_j, \; p_j\big)\,\big(x^i+\frac12\,\pi^{ij}\,p'_j, \; p'_j\big)=\big(x^i+\frac12\,\pi^{ij}\,(p'_j-p_j), \; p_j+p_j'\big)\ .
    \end{array}
    \end{equation}
    Again the canonical symplectic structure $\omega$ obeys the multiplicativity condition
    \begin{align*}
    \dd\big(x^i+\tfrac12\,\pi^{ij}\,(p'_j-p_j)\big)\wedge \dd(p_i+p'_i)
    = \dd\big(x^i-\tfrac12\,\pi^{ij}\,p_j\big)\wedge \dd p_i+ 
    \dd\big(x^i+\tfrac12\,\pi^{ij}\,p'_j\big)\wedge \dd p'_i\ .
    \end{align*}
By construction
$$
\{\sfs^i,\sfs^j\}=\pi^{ij}\quad,\quad \{\sft^i,\sft^j\}=-\pi^{ij}\quad,\quad \{\sfs^i, \sft^j\}=0\ .
$$
In the new coordinates $(x,p)$, where $x^i=\sfs^i(q,p)$, the symplectic form on $\mathbb{R}^{2n}$ is given by $$
\omega=\dd p_i\wedge \dd x^i+\tfrac12\, \pi^{ij}\,\dd p_i\wedge \dd p_j\ .
$$ 
\end{example}

\begin{example} 
The Poisson structure of any symplectic manifold $(M, \sigma)$ is  integrable. As an integrating  symplectic groupoid, one can always take the {\it pair groupoid} $M\times M\rightrightarrows M$ defined by the structure maps
$$
\sfs(x,y)=x\quad,\quad \sft(x,y)=y\quad,\quad (x,y)\,(y,z)=(x,z)\quad,\quad (x,y)^{-1}=(y,x)\ .
$$
The submanifold of units is identified with the diagonal of $M\times M$, which is diffeomorphic to $M$. The pair groupoid is transitive. The compatible symplectic structure on $M\times M$ is given by the difference of pullbacks
$$
\omega=\sfs^\ast(\sigma)-\sft^\ast(\sigma)\ .
$$
By construction, the source (target) projection is a Poisson (anti-Poisson) map. 
\end{example}

\section{Bisections}
\label{sec:bisections}

As discussed in \S\ref{sec:intro}, symplectic groupoids will play the role of phase spaces of  point particles living on spacetimes with Poisson structures. To introduce electromagnetic couplings of these particles we need one more fundamental concept related to groupoids, namely the notion of `bisection', which we now explain and develop in some detail in anticipation of our applications to Poisson electrodynamics.

\paragraph{\underline{\textsf{Bisections of Lie groupoids.}}}

\begin{definition}
    Let $G\rightrightarrows X$ be a Lie groupoid. A submanifold $\Sigma\subset G$ is a \emph{bisection} if 
    $\sfs|_\Sigma $ and $ \sft|_\Sigma$ are diffeomorphisms onto $X$.  
    \end{definition}

To each bisection $\Sigma$ there corresponds a pair of maps $\Sigma_\sfs: X\rightarrow G$ and $\Sigma_\sft: X\rightarrow G$ defined as $\Sigma_\sfs= (\sfs|_\Sigma)^{-1}$ and $\Sigma_\sft=(\sft|_{\Sigma})^{-1}$, and hence the name.  

The set of all bisections $\mathscr{B}(G)$ of $G$ has the  structure  of a regular Lie group for the multiplication operation~\cite{Rybicki2001OnTG}
$$
\Sigma_1\cdot \Sigma_2:=\big\{g_1\,g_2\;\big|\; (g_1,g_2)\in(\Sigma_1\times\Sigma_2)\cap(G\ast G)\big\}\ ,
$$
where $\Sigma=X$ plays the role of the unit and the inverse of $\Sigma$ is $\Sigma^{-1}=\{g^{-1}\;|\;g\in\Sigma\}$. 

The group $\mathscr{B}(G)$  acts naturally from the left and right on $G$ as
\begin{equation}\label{LRT}
L_\Sigma (g)= \sigma_1\, g \qquad\text{and}\qquad R_\Sigma(g)=g\, \sigma_2 \ ,
\end{equation}
where $\sigma_1$ and $\sigma_2$ are the unique elements of $\Sigma$ satisfying   $\sft(\sigma_1)=\sfs(g)$ and $\sft(g)=\sfs(\sigma_2)$, namely $\sigma_1=\Sigma_\sft\circ\sfs(g)$ and $\sigma_2=\Sigma_\sfs\circ\sft(g)$. By associativity, the right and left actions commute. Sometimes we will denote them as $L_\Sigma(g)=\Sigma\, g$ and $R_\Sigma(g)=g\,\Sigma$, see Figure~\ref{Gr2}. Notice that $\Sigma_\sft(x)=\Sigma\, x$ and $\Sigma_\sfs(x)=x\,\Sigma$ for all $x\in X$. It follows from the definition that the right (left) action of $\mathscr{B}(G)$  preserves $\sfs$-fibres ($\sft$-fibres) while taking $\sft$-fibres to $\sft$-fibres ($\sfs$-fibres to $\sfs$-fibres).

\begin{figure}[ht]
\center{
\Large
\begin{tikzpicture}[scale=1.5]

\draw[] (6.5,0) -- (11.5,0);
\draw[-, color=blue] (9.75,-1.5) -- (11.5,2);
\draw[-, color=blue] (7.25,-1.5) -- (9,2);
\draw[-, color=blue] (8.25,-1.5) -- (10,2);
\draw[-, color=blue] (9.25,-1.5) -- (11,2);
\draw[-, color=blue] (6.75,-1.5) -- (8.5,2);
\draw[-, color=blue] (7.75,-1.5) -- (9.5,2);
\draw[-, color=blue] (8.75,-1.5) -- (10.5,2);

\draw[-, color=red] (9.5,2) -- (11.25,-1.5);
\draw[-, color=red] (7,2) -- (8.75,-1.5);
\draw[-, color=red] (8,2) -- (9.75,-1.5);
\draw[-, color=red] (9,2) -- (10.75,-1.5);
\draw[-, color=red] (6.5,2) -- (8.25,-1.5);
\draw[-, color=red] (7.5,2) -- (9.25,-1.5);
\draw[-, color=red] (8.5,2) -- (10.25,-1.5);

\fill [black]  (8.5,0) circle (1.2pt);
\fill [black]  (8,1) circle (1.2pt);
\fill [black]  (8.75,0.5) circle (1.2pt);
\fill [black]  (8.25, 1.5) circle (1.2pt);
\fill [black]  (10,1) circle (1.2pt);
\fill [black]  (9.5,0) circle (1.2pt);
\fill [black]  (10.5,0) circle (1.2pt);

\coordinate [label=right:${}_{{}_{h}}$] (A) at (8.75, 0.6);
\coordinate [label=left:${}_{{}_{g}}$] (A) at (8, 0.93);
\coordinate [label=right:${}_{{}_{\!\Sigma\, h=g\,\Sigma'}}$] (A) at (8.25, 1.5);
\coordinate [label=below:${}_{{}_{\sft\!(\!g\!)=\sfs\!(\!h\!)}}$] (A) at (8.51, 0.08);
\coordinate [label=below:${}_{{}_{x}}$] (A) at (9.5, 0);
\coordinate [label=below:${}_{{}_{\quad r_\Sigma(x)}}$] (A) at (10.5, 0);
\coordinate [label=right:${}_{{}_{x\Sigma }}$] (A) at (10, 1);

\coordinate [label=:${}_{X}$] (A) at (6.75, 0);
\coordinate [label=:${}_{\Sigma}$] (A) at (6.6, 1.1);
\coordinate [label=below:${}_{\Sigma'}$] (A) at (11.2, 0.6);

\draw [] (6.5, 1) ..controls (7.5, 1.9)  and (8.2,-0.2).. (11.5,1.63);
\draw [] (6.5, 0.5) ..controls (8, 1.2)  and (9,-0.15).. (11.3,0.7);
\end{tikzpicture}
}
\caption{\small The left and right actions of bisections.}
\label{Gr2}
\end{figure}
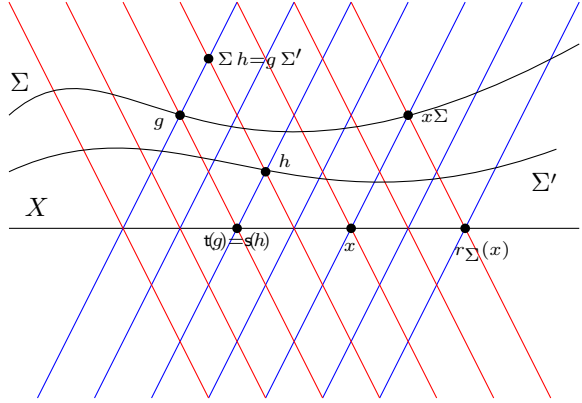

\begin{example}
    In the case of a group $G$ regarded as a groupoid over the identity element (see Example \ref{GG}), bisections are just elements of $G$ and the group $\mathscr{B}(G)$ is isomorphic to $G$. The left and right actions of $\mathscr{B}(G)$ correspond to left and right translations on $G$. From this perspective, the concept of bisection generalises the left and right translations on groups to the case of Lie groupoids. 
\end{example}

The group $\mathscr{B}(G)$ also acts by diffeomorphisms on the base manifold as
\begin{equation}\label{lt}
 l_\Sigma(x)= \sfs\circ\Sigma_\sft(x) \qquad\text{and}\qquad  r_\Sigma(x)=\sft\circ\Sigma_\sfs(x)\ , 
\end{equation}
for all $\Sigma\in\mathscr{B}(G)$ and $x\in X$.
The diffeomorphisms $l_\Sigma$ and $r_\Sigma$ are mutually inverse to each other.  Hence each bisection $\Sigma$ is completely specified by either of $\Sigma_\sfs$ or $\Sigma_\sft$. 

Let $\mathscr{B}(G)^{\mathrm{op}}$ denote the opposite group to  $\mathscr{B}(G)$. Then the formulas (\ref{LRT}) and (\ref{lt}) define group homomorphisms
\begin{align*}
L: \mathscr{B}(G)\longrightarrow \mathsf{Diff}(G)\qquad &\text{and} \qquad l : \mathscr{B}(G)\longrightarrow \mathsf{Diff}(X)\ ,\\[4pt]
R: \mathscr{B}(G)^{\mathrm{op}}\longrightarrow \mathsf{Diff}(G)\qquad &\text{and} \qquad r : \mathscr{B}(G)^{\mathrm{op}}\longrightarrow \mathsf{Diff}(X)\ .
\end{align*}
These left and right actions give the spaces of differential forms $\sfLambda(G)$ and $\sfLambda(X)$ the structures of modules over the groups 
$\mathscr{B}(G)$ and $\mathscr{B}(G)^{\mathrm{op}}$, respectively.  In the following we will need  

\begin{proposition}\label{p1}
For any multiplicative $k$-form $\alpha\in \sfLambda^k (G)$ define the pair of maps 
\begin{align}\label{ST}
\begin{split}
  S_\alpha :\mathscr{B}(G)\longrightarrow \sfLambda^k(X) \ , &  \quad \Sigma\longmapsto S^\Sigma_\alpha =\Sigma_\sfs^\ast\alpha\ ;\\[4pt]
     T_\alpha :\mathscr{B}(G)^{\mathrm{op}}\longrightarrow \sfLambda^k(X)\ , &\quad  \Sigma\longmapsto T^\Sigma_\alpha =\Sigma_\sft^\ast\alpha\ .
     \end{split}
\end{align}
Then
\begin{equation}\label{RL}
    R^\ast_\Sigma \alpha=\alpha +\sft^\ast S^\Sigma_\alpha \qquad\text{and}\qquad L^\ast_\Sigma \alpha=\alpha+\sfs^\ast T^\Sigma_\alpha \ .
\end{equation}
Furthermore, the maps (\ref{ST}) are crossed homomorphisms from the groups of bisections to the corresponding modules:
\begin{equation}\label{CH}
    S^{\Sigma_1\cdot\Sigma_2}_\alpha =S^{\Sigma_1}_\alpha +r^\ast_{\Sigma_1} S^{\Sigma_2}_\alpha \qquad\text{and}\qquad T^{\Sigma_1\cdot\Sigma_2}_\alpha =T^{\Sigma_2}_\alpha +l^\ast_{\Sigma_2} T^{\Sigma_1}_\alpha
\end{equation}
for all $\Sigma_1, \Sigma_2\in \mathscr{B}(G)$. 
\end{proposition}

\begin{proof}
We will prove only the first identities in (\ref{RL}) and (\ref{CH}), leaving the proof of the remaining two to the reader. 
Given a bisection $\Sigma\in\mathscr{B}$, consider the submanifold of composable elements
$$
G\ast \Sigma =\{(g,\sigma)\in G\times\Sigma\;|\; \sft(g)=\sfs(\sigma)\}\subset G\times G\ .
$$
Then $G\ast \Sigma\simeq G$ and one can write each element of $G\ast \Sigma$ as $\big(g, \Sigma_\sfs(\sft(g))\big)$ for some $g\in G$. Hence 
$$
\sfp_1\big(g, \Sigma_\sfs(\sft(g))\big)=g\quad,\quad \sfp_2\big(g,\Sigma_\sfs(\sft(g))\big)=\Sigma_\sfs(\sft(g))\quad,\quad \sfm(g,\sigma)=g\,\sigma =g\,\Sigma\ .
$$
Since $\alpha$ is multiplicative, we can write
\begin{equation}\label{rel}
    R^\ast_\Sigma \alpha=\sfm^\ast\alpha=\sfp^\ast_1\,\alpha+\sfp_2^\ast\,\alpha=\alpha +\sft^\ast\, \Sigma_\sfs^\ast\alpha\ .
\end{equation}
This proves the first relation in (\ref{RL}). 

On substituting $\Sigma=\Sigma_1\cdot\Sigma_2$, we then find 
\begin{align}\label{RR}
\begin{split}
    R^\ast_{\Sigma_1\cdot\Sigma_2} \alpha&=\big(R_{\Sigma_2}\circ R_{\Sigma_1}\big)^\ast \alpha =R^\ast_{\Sigma_1}\,R^\ast_{\Sigma_2} \alpha \\[4pt]
   &= R^\ast_{\Sigma_1}\big(\alpha +\sft^\ast\, (\Sigma_2)_\sfs^\ast\alpha\big)
    =R^{\ast}_{\Sigma_1}\alpha +R^\ast_{\Sigma_1}\,\sft^\ast\, (\Sigma_2)_\sfs^\ast\alpha\ .
    \end{split}
\end{align}
Denote by  $\sfi: X\hookrightarrow G$ the embedding of the base manifold as the submanifold of units $X\subset G$. It obeys the identity
$R_\Sigma \circ \sfi=\Sigma_\sfs$. Applying now the pullback map $\sfi^\ast$ to  (\ref{RR}), we get
$$
(\Sigma_1\cdot\Sigma_2)_\sfs^\ast\alpha=(\Sigma_1)_\sfs^\ast \alpha +r_{\Sigma_1}^\ast \, (\Sigma_2)_\sfs^\ast\alpha
$$
and the first relation in (\ref{CH}) follows.\footnote{It seems that the first proof of these relations was published in \cite[Thm.~7.11]{kosmann2015multiplicativity}.}
\end{proof}

\begin{remark}
Acting with the map $\sfi^\ast$ on both sides of (\ref{rel}) and using the identity $\sft\circ \sfi=\mathrm{id}_X$, we arrive at one more interesting relation: $\alpha|_X=0$. That is, the restriction of any multiplicative form $\alpha$ to the base of a Lie groupoid $G\rightrightarrows X$ vanishes.  
\end{remark}

\paragraph{\underline{\textsf{Lagrangian bisections.}}}

For symplectic groupoids, one can define the subset $\mathscr{L}(G)$ of \emph{Lagrangian bisections}. By definition
\begin{equation*}
    \Sigma\in \mathscr{L}(G)\quad \Longleftrightarrow\quad \omega\big|_\Sigma=0 \ .
\end{equation*}
The Lagrangian bisections form a subgroup of the group $\mathscr{B}(G)$. Furthermore, the left and right translations (\ref{LRT}) are symplectomorphisms of $(G,\omega)$ if and only if the bisection  $\Sigma$ is Lagrangian~\cite[Lem.~9.11]{Visman}: this follows immediately from the relations (\ref{RL}), with $\alpha=\omega$, by noting that the condition $\omega|_\Sigma=0$ is equivalent to either of the two equations $\sft^\ast\, \Sigma^\ast_\sfs\omega=0$ or $\sfs^\ast\, \Sigma^\ast_\sft\omega=0$.

\begin{example}\label{Ex212}
Consider the symplectic groupoid of Example \ref{zero-PB}. This integrates the zero Poisson structure on $X$ and is given by the cotangent bundle $T^\ast X$. Since the source and target maps coincide here with the canonical projection $\sfp: T^\ast X\rightarrow X$, the group $\mathscr{B}(T^*X)$ is just the additive group of sections $\sfGamma(T^\ast X)=\sfLambda^1(X)$, i.e., the space of one-forms on $X$. Then one easily sees that a one-form $A\in \sfLambda^1(X)$ defines a Lagrangian bisection if and only if it is closed, $\dd A=0$. The closed one-forms constitute a subgroup $\sfZ\sfLambda^1(X)$ of the group of all one-forms in accordance with the general inclusion $\mathscr{L}(G)\subset \mathscr{B}(G)$. The coinciding of the source and target maps also implies that the group $\sfLambda^1(X)$ acts trivially on $X$. 
\end{example}

\begin{example}\label{const-PB-BS}
Let us now explore the group of bisections associated with the symplectic groupoid of Example~\ref{const-PB}. Consider a pair of bisections $\Sigma^A$ and $\Sigma^B$ defined in Darboux coordinates by the maps
\begin{equation*}
    \Sigma_\sfs^A(x)=\big(x^i+\tfrac 12\,\pi^{ij}\,A_j(x), A_k(x)\big)\qquad\text{and}\qquad \Sigma^B_\sft(x)=\big(x^i-\tfrac12\,\pi^{ij}\, B_j( x), B_k(x)\big)    
    \end{equation*}
 for some one-forms $A=A_i(x)\,\dd x^i$ and  $B=B_i(x)\,\dd x^i$ on $\mathbb{R}^n$. Here we adapt the parametrization of the bisections $\Sigma^A$ and $\Sigma^B$ by one-forms to the $\sfs$-projection and to the $\sft$-projection, respectively.  According to (\ref{comp}) the product of these bisections is given by  
\begin{equation}\label{ABC1}
   \Sigma^{A'} := \Sigma^B\cdot\Sigma^A= \big(x^i+\tfrac12\,\pi^{ij}\,\big(A_j(x)-B_j(x)\big), A_k(x)+B_k(x)\big)\ .
\end{equation}
Notice that the parametrization of $\Sigma^{A'}$  is adapted to neither the $\sfs$-projection nor the $\sft$-projection. 

Suppose the one-form  $B$ is infinitesimally small so that $\Sigma^B$ is close to the base $X=\mathbb{R}^n$. 
Then by shifting $x^i\rightarrow  x^i+\pi^{ij}\, B_j(x)$ we can  write
\begin{equation}\label{ABC2}
     \Sigma^{A'}=\big (x^i+\tfrac12\,\pi^{ij}\, A'_j(x), A'_k(x) \big)\ ,
     \end{equation}
where 
\begin{equation*}
    A'_i=A_i+B_i+\partial_j A_i\,\pi^{jk}\, B_k+O(B^2)\ .
\end{equation*}
This gives the variation of the one-form $A$ under infinitesimal left translations as
\begin{equation*}
\delta_B A_i=A'_i - A_i= B_i +\partial_j A_i\,\pi^{jk}\, B_k\ . 
\end{equation*}

The condition that the bisection $\Sigma^B$ is Lagrangian takes the form 
\begin{equation*}
    \dd q^i\wedge \dd p_i\big|_{\Sigma^B}=\dd B+O(B^2)=0\ .
\end{equation*}
 Locally it implies that $B=\dd\varepsilon$ for some small function $\varepsilon(x)$. In this way we arrive at the standard infinitesimal gauge transformations
 \begin{equation*}
     \delta_\varepsilon A_i=\partial_i\varepsilon+\{A_i,\varepsilon\}
 \end{equation*}
  for Poisson electrodynamics, cf. (\ref{SW2}). 
  
 By comparing  (\ref{ABC1}) with (\ref{ABC2}) one can also write a formula for finite gauge transformations that integrate the infinitesimal gauge transformations. A straightforward computation yields 
\begin{equation*}
    A'_i(x)=A_i(y)+B_i(y)\,,
\end{equation*}
 where $y=y(x)$ is determined by the equation 
\begin{equation*}
    x^i=y^i-\pi^{ij}\, B_j(y)= l_{\Sigma^B}(y)^i\ .
\end{equation*}
In other words, $y=r_{\Sigma^B}(x)$. 
The condition that $\Sigma^B$ is Lagrangian now takes the form 
\begin{align*}
\big(\Sigma^B\big)_\sft^\ast(\dd q^i\wedge \dd p_i)&=\dd\big(x^i-\tfrac12\, \pi^{ij}\,B_j(x)\big)\wedge \dd B_i(x)\\[4pt]
   &= -\tfrac12\,\big(\partial_i B_j-\partial_j B_i+\pi^{kl}\,\partial_i B_k\,\partial_j B_l\big)\,\dd x^i\wedge \dd x^j=0\ . 
\end{align*}
\end{example}

\begin{example}
Generalizing Example~\ref{const-PB-BS}, we consider infinitesimal left translations on the group of bisections associated with a local symplectic groupoid. In the setting of \S\ref{sec:symplgr}, let $(X,\pi)$ be an $n$-dimensional Poisson manifold with local coordinate chart $(U,q^i)$, and let $W$ be a neighbourhood of the zero section in $T^*U$, parametrized by $(q^i,p_j)$.
    For a pair of composable points $(\tau(x,p_1), p_1)$ and $(\sigma(x,p_2), p_2)$ we can write
    \begin{equation}\label{tss}
        \big(\tau(x,p_1), p_1\big)\,\big(\sigma(x,p_2), p_2\big)=\big(\sigma(x',p_3), p_3\big)\ ,
        \end{equation}
where $x\in U$ while the maps $\tau$ and $\sigma$ are defined by(\ref{ODE}) and (\ref{int}). With these equations, one can see that 
\begin{equation*}
    \tau^i (x, p)=x^i-\tfrac12\,\pi^{ij}(x)\,p_j+O(p^2) \qquad \text{and} \qquad s^i\big(\tau(x,p), p\big)=x^i-\pi^{ij}(x)\,p_j + O(p^2)\ .
\end{equation*}
Equating the sources in the left-hand side and right-hand side of (\ref{tss}), we find immediately
\begin{equation*}
    x'=x-\pi(x)\,p_1+O(p_1^2)\ .
\end{equation*}

To determine $p_3$, consider the Hamiltonian $h_\xi(x)=\sfs^i(q,p)\,\xi_i$, where   $\xi\in \mathbb{R}^n$ is a vector-valued parameter. It generates the Hamiltonian flow $\varphi_t : W\rightarrow W$ with respect to the Poisson brackets (\ref{PBr}) given by
\begin{equation*}
    \frac{\dd q^i}{\dd t}=\{h_\xi, q^i\}=-\frac{\partial \,\sfs^j}{\partial p_i}\,\xi_j\qquad\text{and}\qquad   \frac{\dd p_i}{\dd t}=\{h_\xi, p_i\}=\frac{\partial \,\sfs^j}{\partial q^i}\,\xi_j\ .
\end{equation*}
For small $\xi$ we find that the map $(q(1), p(1)) =\varphi_1(q(0), p(0))$ is given by 
\begin{align*}
q^i(1)&=q^i(0)-\xi_j\,\left(\frac{\partial\, \sfs^j}{\partial p_i}\right)\big(q(0),p(0)\big)+O(\xi^2)\ ,\\[4pt]
    p_i(1)&=p_i(0)+\xi_j\,\left(\frac{\partial\, \sfs^j}{\partial q^i}\right)\big(q(0),p(0)\big)+ O(\xi^2)\ .
\end{align*}
Hence setting  $\xi=p_1$ and using the identity 
$$
\left(\frac{\partial \,\sfs^j}{\partial p_i}\right)(q,0)=\frac12\,\pi^{ij}(q)\ ,
$$
we obtain $\varphi_1(x, 0)=(\tau(x,p_1),p_1)$ for any infinitesimally small $p_1$.  Applying the formula (\ref{ab}) then yields 
\begin{align}\label{pr1}
\begin{split}
\big(\tau(x,p_1), p_1\big)\,\big(\sigma(x,p_2), p_2\big) &= \varphi_1 \big(\sigma(x,p_2), p_2\big)\\[4pt] &= \Big(\sigma \big(x-\pi(x)\,p_1, p_2+\gamma(x, p_2)\,p_1\big)\,,\, p_2+\gamma(x, p_2)\,p_1\Big)+O(p_1^2)\ ,
\end{split}
\end{align}
 where 
\begin{equation*}
    \gamma^i_j(x,p)=\left(\frac{\partial \,\sfs^i}{\partial q^j}\right)\big(\sigma(x,p),p\big)\ .
\end{equation*}

 Using (\ref{pr1}), one can compute the product of two (local) bisections $\Sigma^B=\big(\tau(x, B(x)), B(x)\big)$ and $\Sigma^A=\big(\sigma(x,A(x)), A(x)\big)$ whenever $B(x)$ is infinitesimally small:
\begin{align*}
    \Sigma^B\cdot \Sigma^A=: \Sigma^{A'}=\big(\sigma(x, A'(x)), A'(x)\big)
    =\Big(\sigma \big(x-\pi(x)\, B, A+\gamma(x, A)\,B\big)\,,\, A+\gamma(x, A)\,B\Big)\ .
    \end{align*}
By making the shift $x^i\rightarrow x^i+\pi^{ij}(x)\, A_j$, we finally get 
\begin{equation}\label{CBA}
    A'_i=A_i+\gamma_i^j(x, A)\,B_j+\partial_k A_i\, \pi^{kj}(x)\, B_j\ .
\end{equation}
The section $\Sigma^B$ is Lagrangian if and only if $\dd B=0$ or, in our local setting, equivalently $B=\dd\varepsilon$ for some infinitesimal function $\varepsilon\in C^\infty(U)$. Then 
\begin{align*}
    \delta_\varepsilon A_i=A'_i-A_i=\gamma_i^j(x, A)\,\partial_j\varepsilon +\partial_k A_i\, \pi^{kj}(x)\, \partial_j\varepsilon
    \end{align*}
and we arrive at the infinitesimal gauge transformations (\ref{gtr}). 
\end{example}

Generally, the Lie algebra of the group of Lagrangian bisections $\mathscr{L}(G)$ of a symplectic groupoid $G\rightrightarrows X$ is isomorphic to the Lie algebra of closed one-forms $\sfZ\sfLambda^1 (X)$, see \cite[Thm. 4.5]{Rybicki2001OnTG}. The Lie bracket is given by
\begin{align*}
    [B_1,B_2]=\dd\big(\pi(B_1, B_2)\big) \ ,
\end{align*}
for all $B_1, B_2\in \sfZ\sfLambda^1(X)$.
The exact one-forms $\dd\sfLambda^0(X)$ constitute an ideal in $\sfZ\sfLambda^1(X)$ and the quotient $\sfZ\sfLambda^1(X)/\dd\sfLambda^0(X)=\sfH^1(X,\mathbb{R})$ is an abelian Lie algebra. The exterior differential  defines a Lie algebra homomorphism $\dd: \sfLambda^0(X)\rightarrow \dd\sfLambda^0(X)$ if one regards $\sfLambda^0(X)$ as the Lie algebra of functions with respect to the Poisson bracket $\{-,-\}$ of $(X,\pi)$. For a connected base $X$,  $\ker(\dd) = \sfZ\sfLambda^0(X)=\mathbb{R}$. Among other things, this implies that the commutator of two  gauge transformations (\ref{gtr}) is given by the formula (\ref{commut}). 

\section{Classical Poisson electrodynamics}
\label{sec:CPE}

Now we have all the necessary mathematical tools to formulate the dynamics of a point particle coupled to an external electromagnetic field on a spacetime which is a Poisson manifold $(X,\pi)$. For this, we identify the phase space of the particle with an integrating symplectic groupoid $G\rightrightarrows X$. We identify the group of bisections $\mathscr{B}(G)$ with the space of gauge potentials and the subgroup of Lagrangian bisections $\mathscr{L}(G)$ with the gauge group acting on gauge potentials by left translations:
$$
\Sigma \longmapsto \Sigma' \cdot \Sigma \ ,
$$
for all $\Sigma\in \mathscr{B}(G)$ and $ \Sigma' \in \mathscr{L}(G)$. In other words, the physical degrees of freedom of the electromagnetic field are identified with the elements of the left coset $\mathscr{B}(G)/\mathscr{L}(G)$. 

This identification follows from the analogy with ordinary spacetime, see Example~\ref{Ex212}: In the case $\pi=0$, the symplectic groupoid is just the cotangent bundle $T^\ast X$, the space of bisections coincides with the space of one-forms on $X$, and the space of Lagrangian bisections coincides with the space of closed one-forms on $X$. 
Altogether, the main ingredients are summarised in the glossary of Table~\ref{tab:glossary}. Below we explain these ingredients in more detail. 

\paragraph{\underline{\textsf{Minimal coupling.}}}

According to the first two items of Table~\ref{tab:glossary}, the Hamiltonian action functional of a free relativistic  particle has the form 
\begin{equation}\label{S}
    S_{\rm part}=\int_\gamma\, \theta -\lambda\, H\ ,
\end{equation}
where 
\begin{itemize}
    \item $\gamma: \mathbb{R}\rightarrow G$  is a phase space trajectory of the particle; 
    \item $\theta$ is a (locally defined) symplectic potential for the symplectic two-form $\omega$ on $G$, i.e., $\omega =\dd\theta$;
    \item $H$ is a smooth function on $G$;
    \item $\lambda$ is a one-form on $\gamma$ playing the role of the Lagrange multiplier for the Hamiltonian constraint $H\approx 0$ that expresses the strong conservation of the particle's mass.\footnote{For a free scalar particle of mass $m$ in Minkowski space, the Hamiltonian constraint $H=p^2-m^2\approx 0$ expresses the relativistic energy-momentum relation (in units where the speed of light is $c=1$), or equivalently the off-shell conservation of mass.}
    \end{itemize}
The action functional (\ref{S}) enjoys the standard gauge symmetry generated by the first class constraint $H\approx 0$. The Hamiltonian reduction by this constraint gives the physical phase space $\mathcal{P}=G/\!\!/ H$. 

This is reminescent of the approach of {\it elementary dynamical systems} \cite{Souriau}, in which the physical phase space $\mathcal{P}$ 
is postulated to be a homogeneous symplectic manifold for a group of fundamental symmetries $\Gamma$. For example, in the case of a free scalar particle in Minkowski space $\mathbb{R}^{3,1}$, the group $\Gamma$ is the Poincar\'e group acting by canonical transformations in a six-dimensional phase space. The phase space of a massive spinning particle is eight-dimensional \cite{Lyakhovich1996}. For coloured particles one takes $\Gamma$ to be the direct product of the Poincar\'e group with a group of internal symmetries, which leads to  further enlarging of the physical phase space \cite{DH}.  As is well-known, all homogeneous symplectic manifolds are  
exhausted (up to coverings) by the coadjoint orbits of the corresponding symmetry groups. This gives a complete classification of all elementary particles associated with a given group of fundamental symmetries $\Gamma$. 

Irrespective of global symmetries and a particular form of the Hamiltonian $H$, the model (\ref{S}) describes the correct number of physical degrees of freedom for a scalar particle on $X$. If the particle is electrically charged, the next issue is introducing its coupling to an external electromagnetic field. According to the third item of Table~\ref{tab:glossary}, the electromagnetic field is described by bisections $\Sigma\in\mathscr{B}(G)$ of the symplectic groupoid $G\rightrightarrows X$. 
The minimal interaction with the electromagnetic field is introduced by replacing the `free' Hamiltonian $H$ in (\ref{S}) with $H^\Sigma=R^\ast_\Sigma H$,  so that the action functional takes the form \begin{equation}\label{Ssigma}
    S_{\rm part}^\Sigma=\int_\gamma\, \theta-\lambda\, H^\Sigma\ .
\end{equation}

By construction, the Hamiltonian $H^\Sigma$ is a $\mathscr{B}(G)$-equivariant function on the phase space $G$; in particular, $H^{\Sigma'\cdot\Sigma}=R^\ast_{\Sigma'}H^\Sigma$ for all $\Sigma'\in \mathscr{L}(G)$. 
This allows one to compensate the gauge transformation of the electromagnetic field $\Sigma \rightarrow \Sigma'\cdot\Sigma$ by the canonical transformation $R^{-1}_{\Sigma'}: G\rightarrow G$ of the phase space. Since the latter modifies the symplectic potential $\theta$ by nothing more than an exact form, the action (\ref{Ssigma}) remains unchanged under such combined transformations of the gauge potentials and the phase space variables. 

\paragraph{\underline{\textsf{Field strength tensors.}}}

Having identified electromagnetic potentials with  the group of bisections $\mathscr{B}(G)$ and gauge parameters with  the subgroup of  Lagrangian bisections $\mathscr{L}(G)$, we would now like to introduce an object that measures the deviation of a bisection from being Lagrangian. 
In the case of conventional electrodynamics, this role is played by the field strength tensor $F=F(A)=\dd A$ for a one-form $A$ on $X$. For later comparison, we list its main properties:
\begin{enumerate}
\item[F1.] $F$ is a two-form on the spacetime manifold $X$;
\item[F2.] $F$ is closed, $\dd F=0$;
\item [F3.] $F$ is a local function of the gauge potential $A$, i.e. it depends on first order derivatives of $A$;
\item[F4.] $F$ is an additive function of  potentials: $F(A_1+A_2)= F(A_1)+F(A_2)$;
\item [F5.] $F$ is invariant under the gauge transformations $A\rightarrow A+\dd\varepsilon$ for functions $\varepsilon$ on $X$;
\item  [F6.] $F$ vanishes on pure gauge potentials, $F(\dd\varepsilon)=0$.
\end{enumerate}

It is desirable that the symplectic groupoid counterpart of the electromagnetic field strength tensor enjoys as many of these properties as possible. 
The natural candidates to start with are given by the expressions
\begin{align*}
    F^\sfs=F^\sfs(\Sigma)=S^\Sigma_\omega=\Sigma_\sfs^\ast\omega\qquad\text{and}\qquad F^\sft=F^\sft(\Sigma)=T^\Sigma_\omega=\Sigma_\sft^\ast\omega\ ,
\end{align*}
for any bisection $\Sigma\in\mathscr{B}(G)$.
Both quantities $F^\sfs$ and $F^\sft$ are closed two-forms on $X$, so they satisfy properties F1 and F2. Property~F6 also holds as the two-forms $F^\sfs$ and $F^\sft$ vanish on Lagrangian bisections by definition. 

Here the role of the additive group of gauge potentials is played by the (non-abelian) group of bisections $\mathscr{B}(G)$. By Proposition~\ref{p1},
$$
F^\sfs(\Sigma_1\cdot \Sigma_2)=F^\sfs(\Sigma_1) +r^\ast_{\Sigma_1}F^\sfs(\Sigma_2)\qquad\text{and}\qquad F^\sft(\Sigma_1\cdot \Sigma_2)=F^\sft(\Sigma_2) +l_{\Sigma_2}^\ast F^\sft(\Sigma_1)\ .
$$
These relations generalize the additivity property F4. In particular, for a pure gauge $\Sigma_1\in \mathscr{L}(G)$ we obtain 
$$
F^\sfs(\Sigma_1\cdot \Sigma_2)=r^\ast_{\Sigma_1}F^\sfs(\Sigma_2)\qquad\text{and}\qquad F^\sft(\Sigma_1 \cdot \Sigma_2)=F^\sft(\Sigma_2)\ .
$$
Hence the two-form $F^\sft$ is gauge invariant,  while $F^\sfs$ is only gauge covariant.
The two field strength tensors are related by the identities
\begin{equation}\label{ts}
    F^\sfs=r^\ast_\Sigma F^\sft\qquad\text{and}\qquad F^\sft=l^\ast_\Sigma F^\sfs\ .
\end{equation}
For example, 
$$
r^\ast_\Sigma F^\sft=r^\ast_\Sigma\, \Sigma^\ast_\sft\omega=(\Sigma_\sft\circ r_\Sigma)^\ast \omega =(\Sigma_\sft \circ \sft \circ \Sigma_\sfs)^\ast \omega=\Sigma_\sfs^\ast\omega=F^\sfs\ .
$$
It follows from (\ref{ts}) that whenever one of the field strength tensors is local (in the sense of property~F3) the other is not. 

\begin{remark}\label{r31}
The two-form $F^\sfs$ defines an equivariant map $F^\sfs: \mathscr{B}(G)\rightarrow \sfLambda^2(X)$ for the group of Lagrangian bisections $\mathscr{L}(G)$. 
Another simple way to assign such an equivariant map to every tensor field  $\mathcal{T}$ on $X$ is via
\begin{align*}
    \mathcal{T}^\Sigma=r^\ast_{\Sigma}\mathcal{T}\ .
\end{align*}
Then $\mathcal{T}^{\Sigma'\cdot\Sigma}=r_{\Sigma'}^\ast \mathcal{T}^\Sigma$ for any (not necessarily Lagrangian) bisection $\Sigma'$. Hence the map 
$\Sigma\mapsto \mathcal{T}^\Sigma$ is equivariant with respect to the left translations of $\mathscr{B}(G)$. The equivariant tensor fields form a closed subalgebra in the full tensor algebra\footnote{By a tensor algebra we understand the space of tensor fields endowed with the tensor product {\it and} the operations of permutation as well as contraction of indices.} of $X$. 
\end{remark}

Using Remark~\ref{r31}, one can introduce one more type of field strength tensor, which we call a {\it framed field strength tensor}. 
Let $\{e^a\}$ be a collection of one-forms on $X$ that constitute a coframe in $T^\ast U$ for some open domain $U\subset X$.  Then each form $e^a_\Sigma=r_\Sigma^\ast e^a$ defines a $\mathscr{B}(G)$-equivariant map from $\mathscr{B}(G)$ to $\sfLambda^1(X)$. Let $e_a^\Sigma$ denote the dual frame of 
$\mathscr{B}(G)$-equivariant vector fields on $U$. 
We  define the framed field strength tensor as 
$$
F_{ab}=F_{ab}(\Sigma) = F^\sfs(\Sigma)\big(e_a^\Sigma, e_b^\Sigma\big)\ .
$$
By construction, $F_{ab}$ is a collection of scalar functions on $X$ that form an anti-symmetric matrix $F_{ab}=-F_{ba}$. Each function defines an 
$\mathscr{L}(G)$-equivariant map: 
$$
F_{ab}(\Sigma'\cdot\Sigma)=r_{\Sigma'}^\ast F_{ab}(\Sigma) \ ,
$$
for all $\Sigma'\in \mathscr{L}(G)$. We also have $F^\sfs=F_{ab}\,e^a_\Sigma\wedge e^b_\Sigma$. 

\begin{example} \label{ex:Fconstpi}
Let us consider the constant Poisson structure $\pi$ of Examples \ref{const-PB} and \ref{const-PB-BS}.  We define a bisection $\Sigma$ by
\begin{align*}
    \Sigma_\sfs(x)=\big(x^i+\tfrac12\,\pi^{ij}\, A_j(x), A_k(x)\big)\ ,
\end{align*}
for a vector-valued function $A$ on $\mathbb{R}^n$. Then the right action of $\mathscr{B}(G)$ on $X=\mathbb{R}^n$ takes the form
     \begin{align*}
         r_\Sigma (x)^i=x^i+\pi^{ij}\,A_j(x)\ .
     \end{align*}
The equivariant field strength is given by 
  \begin{align*}
     F^\sfs= \Sigma^\ast_\sfs(\dd p_i\wedge \dd q^i)=\dd A_i\wedge \dd x^i +\tfrac12\, \pi^{ij}\,\dd A_i\wedge \dd A_j=\tfrac12\, F^\sfs_{ij}\,\dd x^i\wedge \dd x^j\ ,
     \end{align*}
     where
     $$
    F_{ij}^\sfs =\partial_iA_j-\partial_jA_i+\pi^{kl}\,\partial_iA_k\,\partial_jA_l\ .
     $$

 Applying the right translations to the basis one-forms $\dd x^i$,  we define the $\mathscr{B}(G)$-equivariant coframe 
  \begin{equation}\label{fr-om}
\omega^a=\omega^a_\Sigma=r^\ast_\Sigma\, \dd x^a= \dd r_\Sigma(x)^a=(\delta^a_i+\pi^{aj}\,\partial_iA_j)\,\dd x^i
  \end{equation}
in the cotangent bundle $T^*X$. Similarly, the tangent bundle $TX$ can be equipped with the $\mathscr{L}(G)$-equivariant frame of Poisson vector fields
  \begin{equation}\label{v}
      v_a=v_a^\Sigma=\big(\delta_a^i+\partial_jA_a\,\pi^{ji}\big)\,\frac{\partial}{\partial x^i}=\frac{\partial}{\partial x^a}+\{A_a,-\}\ .
  \end{equation}
Using the representation 
$ \pi^{ab}\,v_b=\{r^\ast_\Sigma(x^a), - \}$,
one can easily see that 
 \begin{align*}
     v_a^{\Sigma'\cdot\Sigma}=r^\ast_{\Sigma'}\big(v_a^{\Sigma} \big)\ ,
\end{align*}
for all $\Sigma' \in \mathscr{L}(G)$. 

The Lie brackets of frame fields are the Hamiltonian vector fields 
 \begin{align*}
     [v_a,v_b]=\{ \widehat F_{ab}, -\}\ ,
 \end{align*}
 whose Hamiltonians form the antisymmetric matrix
\begin{equation}\label{hatF}
    \widehat F_{ab}=\widehat F_{ab}(\Sigma)=\partial_a A_b-\partial_b A_a+\{A_a,A_b\}\ .
    \end{equation}
    It is this matrix that is usually taken as the standard `field strength tensor' of Poisson electrodynamics associated with a constant Poisson bivector, cf. (\ref{F}). By construction, the matrix elements are $\mathscr{L}(G)$-equivariant functions, i.e., $\widehat F_{ab}(\Sigma'\cdot\Sigma)=r^\ast_{\Sigma'}\widehat F_{ab}(\Sigma)$ for all $\Sigma'\in \mathscr{L}(G)$. 
    
    The frame $v_a$ and coframe $\omega^a$ are not canonically dual to each other. Instead one has
 \begin{equation}\label{vw}
     \omega^a(v_b)=\omega^a_i\,v^i_b= \delta^a_b- \widehat F_{bc}\,\pi^{ca}\qquad\text{and}\qquad \omega_i^a\, v_a^j=\delta^j_i-\pi^{jk}F^\sfs_{ki}\ .
 \end{equation}
 One can verify the identity
 \begin{equation*}
 F^\sfs_{ij}\,v^j_b=\omega^a_i\, \widehat F_{ab}\ .
 \end{equation*}
Now the framed field strength tensor associated with  (\ref{v}) is given by 
\begin{equation*}
  F_{ab}=v^i_a\, F^\sfs_{ij}\,v^j_b=\widehat F_{ab}-\widehat F_{ak} \, \pi^{kl} \, \hat F_{lb}\ .
  \end{equation*}
  With the help of (\ref{vw}) one can also find
  \begin{equation*}
      \omega^a_i\,\omega^b_j\, \widehat F_{ab}=F^\sfs_{ij}-F^\sfs_{ik}\,\pi^{kl}\,F^\sfs_{lj} \ .
  \end{equation*}
Finally, introducing the dual vector fields  \smash{$\widecheck \omega_a=\widecheck \omega{}_a^i\,\frac\partial{\partial x^i}$} to the coframe of one-forms (\ref{fr-om}), that is, \smash{$\omega^a(\widecheck \omega_b)=\delta_b^a$}, we can express  the components of the standard field strength through the covariant field strength tensor as
\begin{equation*}
       \widehat F_{ab}=\widecheck \omega{}^i_a\,\big(F^\sfs_{ij}-F^\sfs_{ik}\,\pi^{kl}\,F^\sfs_{lj}\big)\,\widecheck \omega{}^j_b\ .
       \end{equation*}
Both the covariant and the standard field strength tensors are local and vanish simultaneously. 
\end{example}

\paragraph{\underline{\textsf{Action functional for the electromagnetic field.}}} 

To describe the dynamics of the whole system 
$$
\mbox{charged particle \ + \ electromagnetic field}
$$
the particle action functional (\ref{Ssigma}) should be augmented by the action functional of the electromagnetic field itself.  Although a detailed discussion of the field theoretical aspects of Poisson electrodynamics is beyond the scope of this paper, we will briefly indicate some basic constructions here. 

As a main building block for the action functional, one can use 
 any of the field strength tensors $F$ discussed above. Additional ingredients for the construction are a background metric $g$ on $X$, a volume form $\mu_X$, the Poisson bivector $\pi$, and any other external fields. The main requirement is that the corresponding action functional 
\begin{equation*}
    S_{\rm em}=\int_X\, L(F, g, \pi, \ldots) \ \mu_X
\end{equation*}
be gauge invariant. 

The simplest way to satisfy the gauge invariance condition is to use the gauge invariant field strength tensor $F^\sft$.   Another option is to use the covariant field strength tensor $F^\sfs$ which, unlike $F^\sft$, is always local.  As explained in Remark \ref{r31}, one can make  every tensor field on $X$, such as the metric $g$, into an $\mathscr{L}(G)$-equivariant tensor field by applying right translations, e.g.  $g\rightarrow g^\Sigma=r^\ast_\Sigma g$. With the $\mathscr{L}(G)$-equivariant tensor fields $F^\sfs$, $g^\Sigma$, $\mu_X^\Sigma$, $\pi$, $\dots$ at hand, one can  then produce an equivariant Lagrangian $L=L(F^\sfs, g^\Sigma, \pi, \ldots)$ and define the gauge invariant action functional as 
\begin{equation}\label{SS}
    S_{\rm em}=\int_X\, L\big(F^\sfs, g^\Sigma, \pi, \ldots\big) \ \mu_X^\Sigma\ .
\end{equation}
Both options in fact lead to the same action functional whenever the Lagrangian does not depend on $\pi$.\footnote{The bivector $\widecheck {F}= \pi^\Sigma-\pi$ provides one more type of $\mathscr{L}(G)$-equivariant field strength tensor,  which could be called a {\it contravariant field strength tensor}. Unlike $F^\sfs$ and $F^\sft$, the contravariant field strength tensor  vanishes for $\pi=0$. For a constant Poisson structure, \smash{$\widecheck{F}{}^{ij}=\pi^{ia}\,\pi^{jb}\,\widehat F_{ab}$}.} The point is that the corresponding 
integrands are related by the diffeomorphism:
\begin{equation*}
    L(F^\sfs, g^\Sigma, \ldots ) \ \mu_X^\Sigma=r^\ast_\Sigma \big (L(F^\sft, g, \ldots) \ \mu_X \big)\ .
\end{equation*}

Besides the Poisson bivector $\pi$, there may exist some other $\mathscr{L}(G)$-invariant tensor fields on $X$, for example an invariant volume form $\mu_X$. For instance, if  $\pi$ is invertible then one can take $\mu_X$ to be the  canonical volume form on the symplectic manifold $(X, \pi^{-1})$. Replacing $\mu_X^\Sigma$ with an invariant volume form $\mu_X$ in (\ref{SS}) yields another gauge invariant action functional.  As is well known, not every Poisson manifold admits a volume form invariant under all Poisson diffeomorphisms: the obstruction is represented by the {\it modular class} of the Poisson manifold \cite{WEINSTEIN1997379}. The modular class vanishes for symplectic manifolds and for constant Poisson structures on $\mathbb{R}^n$. For the linear Poisson bracket on $\mathfrak{g}^\ast$, the modular class is zero if and only if the Lie algebra $\mathfrak{g}$ is unimodular. 
It should be emphasized that, regardless of the existence of an invariant volume form on $X$, one can always define an $\mathscr{L}(G)$-equivariant  form $\mu_X^\Sigma$, which is enough  for the purpose of constructing the gauge invariant action functional (\ref{SS}).   
\begin{example}
Consider $X=\mathbb{R}^n$ endowed with a constant Poisson bivector $\pi$ from Example~\ref{ex:Fconstpi}. Using the $\mathscr{B}(G)$-equivariant coframe fields \eqref{fr-om}, the simplest Lagrangian takes the form $L=g^{ij}\,g^{kl}\, F^\sfs_{ik}\,F^\sfs_{jl}$, where
\begin{equation*}
    g_{ij}=\eta^\Sigma_{ij}=\eta_{ab}\,\omega^a_i\,\omega^b_j=\eta_{ab}\,\big(\delta^a_i+\pi^{ak}\,\partial_i A_k\big)\,\big(\delta^b_j+\pi^{bl}\partial_j A_l\big)
\end{equation*}
and $\eta_{ab}$ is the Minkowski metric. One can take $\mu_X=\dd^nx=\dd x^1\wedge \cdots\wedge \dd x^n$ as an invariant volume form.\footnote{Another possibility is to put $\mu_X=\omega^1\wedge \cdots \wedge \omega^n=\sqrt{\det (g)} \ \dd^nx$. Unlike $\dd^nx$, this volume form depends on the electromagnetic field.} This Lagrangian reduces to the conventional Lagrangian of Maxwell's electrodynamics when $\pi= 0$.  However, a detailed comparison with the Lagrangian of noncommutative Maxwell theory suggests a  more complex extension given by
\begin{equation*}
    S[A]=\int_X\,\eta^{ab}\,\eta^{cd}\,\widehat F_{ac}\,\widehat F_{bd} \ \dd^nx 
=\int_X\,g^{ik}\,g^{jl}\,\big(F^\sfs_{ij}+ F^\sfs_{ir}\,\pi^{rs}\,F^\sfs_{sj}\big)\,\big(F^\sfs_{kl}+ F^\sfs_{kt}\,\pi^{tu}\,F^\sfs_{ul}\big) \ \dd^nx \ .
\end{equation*}

One further interesting possibility is to set 
\begin{equation*}
    S_\alpha [A]=\int_X\,  \sqrt{\det\big(g^\Sigma+\alpha\, F^\sfs\big)} \ \dd^nx \ ,
\end{equation*}
where $g$ is an arbitrary metric on $X$ and $\alpha\in\mathbb{R}$ is a constant. This may be regarded as a Poisson version of Born--Infeld electrodynamics.
For $\alpha =0$ the Lagrangian is a total divergence.  
\end{example}

\section{Examples}
\label{sec:ex}

We shall now flush out some detailed examples of the dynamics of charged particles coupled to an external electromagnetic field in Poisson electrodynamics.

\subsection{Constant Poisson structures}

Let $X=\mathbb{R}^n$ endowed with a constant Poisson bivector $\pi$, see Examples~\ref{const-PB}, \ref{const-PB-BS} and~\ref{ex:Fconstpi}. The minimal coupling is described by the Lagrangian
\begin{equation}\label{Ex1L}
    L=p_i\,\dot x^i -\tfrac12\,\pi^{ij}\,p_i\,\dot p_j -{\lambda}\,H(X,P)\ , 
\end{equation}
where an overdot denotes the time derivative, $\lambda$ is a Lagrange multiplier and $H$ is a function of the gauge invariant variables 
\begin{equation*}
  X^i=x^i+\tfrac12\,\pi^{ij}\,\big (A_j(x)-p_j\big)\qquad\text{and}\qquad   P_i=p_i+A_i(x)\ .
\end{equation*}
The infinitesimal gauge transformations 
\begin{equation*}
   \delta_\varepsilon A_i=\partial_i\varepsilon+\partial_j A_i\,\pi^{jk}\,\partial_k\varepsilon\quad,\quad \delta_\varepsilon x^i=-\pi^{ij}\,\partial_j \varepsilon\quad,\quad \delta_\varepsilon p_i=-\partial_i\varepsilon 
\end{equation*}
change the Lagrangian (\ref{Ex1L}) by a total derivative. 

For a scalar particle of mass $m$ in Minkowski space, it is natural to choose 
\begin{equation*}
    H=\eta^{ij}\,P_i\, P_j-m^2\ ,
\end{equation*}
where $\eta^{ij}$ is the (mostly minus) Minkowski metric on $\mathbb{R}^n$. Then the equations of motion read
\begin{equation}\label{EoM}
P_i\,P^i=m^2\quad,\quad \dot x^i=2\,\lambda\, v^i_j\, P^j\quad,\quad 
\dot P_i=-2\,\lambda\, \widehat F_{ij}\, P^j\ ,
\end{equation}
where
\begin{equation}\label{vF}
    v^i_j=\delta_j^i+\partial_k A_j\,\pi^{ki}\qquad\text{and}\qquad \widehat F_{ij}=\partial_i A_j-\partial_j A_i+
    \pi^{kl}\,\partial_k A_i\,\partial_l A_j\ ,
\end{equation}
and the indices are raised using the Minkowski metric $\eta^{ij}$. The first equation in (\ref{vF}) defines an equivariant frame of Poisson vector fields $v_j$, while the second equation is given by the components of the standard field strength (\ref{hatF}). 

Let $\widecheck v$ denote the inverse matrix to  $v$, i.e., $\widecheck v{}^i_k\, v^k_j=\delta^i_j$. Then we can solve the first two equations in (\ref{EoM}) for $\lambda$ and $P$ as 
\begin{equation}\label{lp}
\lambda^2=\frac1{4 m^2}\,g_{ij}\,\dot x^i\,\dot x^j\qquad\text{and}\qquad P^j=\frac1{2\lambda}\,\widecheck v{}^j_i\, \dot x^i\ .
\end{equation}
Here the symmetric tensor
\begin{equation*}
    g_{ij}=\widecheck v{}_i^k\,\eta_{kl}\,\widecheck v{}^l_j 
\end{equation*}
may be regarded as an effective  metric induced by the electromagnetic field; in this interpretation, the one-forms $\widecheck v{}^i$ play the role of a vielbein.  We can simplify these expressions by imposing the proper time gauge $\lambda=\frac12$ to fix the reparametrization invariance of the equations of motion. 

On substituting  (\ref{lp}) into the third equation of (\ref{EoM}), we finally obtain the second order differential equations for the particle's worldlines given by
\begin{equation}\label{GS}
     \ddot x^i+\widetilde{\Gamma}^i_{jk}\,\dot x^j\,\dot x^k=-g^{il}\,\widecheck v{}^r_l\,\widehat F_{rs}\, \widecheck v{}^s_k\,\dot x^k \ .
    \end{equation}
    Here the affine connection $\widetilde{\Gamma}^i_{jk}=\Gamma^i_{jk}+S^i_{jk}$ is given by the sum of the Levi--Civita connection
    \begin{equation*}
     \Gamma_{jk}^i=\tfrac12\, g^{il}\,\big(\partial_j g_{lk}+\partial_k g_{lj}-\partial_l g_{jk}\big)
    \end{equation*}
    for the effective metric and the tensor 
    \begin{equation*}
        S^i_{jk}=\tfrac12\,g^{il}\, \widecheck v{}^p_l\,\big(\widecheck v{}^q_k\,g_{jr} + \widecheck v{}^q_j\, g_{kr}\big)\, \pi^{rs}\,\partial_s \widehat F_{pq} \ .   
        \end{equation*}
Notice that  $\widehat F_{pq}$ is a collection of scalar functions and so applying a partial derivative to it is a legitimate tensor operation.  

The appearance of the induced metric and connection is a known phenomenon in noncommutative gauge theory called {\it emergent gravity}, see e.g.~\cite{Langmann:2001yr,rivelles2003noncommutative,Yang:2004vd,Szabo:2006wx,steinacker2010}. Note that the effective connection $\widetilde \Gamma$ is not a metric connection unless the electromagnetic field is homogeneous.  The right-hand side of (\ref{GS}) is determined by the effective  Lorentz force, which is linear in velocity. 

\subsection{Linear Poisson brackets}\label{LPBr}

Any linear Poisson bracket gives rise to the symplectic groupoid of Example \ref{TG}. 
By using the trivialization  (\ref{triv}), we can identify the bisections of $T^\ast G$ with the graphs $\Sigma^A =\big(A(v), v\big)\in G\times \mathfrak{g}^\ast$ of functions $A:\mathfrak{g}^*\to G$. 
Then the right translation by $\Sigma^A$ is defined by the formula
\begin{equation*}
    \big(g, \mathsf{Ad}_g^\ast(v)\big)\,\Sigma^A= \big(g, \mathsf{Ad}_g^\ast(v)\big)\,\big(A(v),v\big)= \big(g\,A(v), \mathsf{Ad}^\ast_{g\,A(v)}\,\mathsf{Ad}^\ast_{A(v)^{-1}}(v)\big)\ .
    \end{equation*}
It is convenient to make the change of the variables $(g,v)=\big(g,\mathsf{Ad}^\ast_g(u)\big)$. In the new variables $(g,u)$ the symplectic potential assumes the form
\begin{equation}\label{tug}
    \theta=\langle u, g^{-1}\,\dd g\rangle\ ,
\end{equation}
while the action of right translations can be written as 
\begin{equation}\label{guh}
    (g,u)\,\Sigma^A=\big(g\,A(u), \mathsf{Ad}^\ast_{A(u)^{-1}}(u)\big)\ .
\end{equation}

The function $A:\mathfrak{g}^\ast\rightarrow G$ is identified with a gauge potential. On substituting (\ref{guh}) into (\ref{tug}) one can see that the bisection $\Sigma^A=\big(A(u),u\big)$ is Lagrangian if and only if 
\begin{equation*}
F^\sfs   = \dd\langle u, \dd A \, A^{-1}\rangle =0\ .
\end{equation*}
The gauge potentials form a non-abelian group isomorphic to the group of bisections $\mathscr{B}(T^\ast G)$. If $\Sigma^B=\big(B(v), \mathsf{Ad}^\ast_{B(v)}(v)\big)\in G\times \mathfrak{g}^\ast$ is another bisection, then
\begin{equation}\label{fh}
    \Sigma^B \cdot \Sigma^A=\big(B(v), \mathsf{Ad}^\ast_{B(v)}(v)\big)\,\big(A(v),v\big) =\big(B(v)\,A(v), \mathsf{Ad}^\ast_{B(v)}(v)\big)\ .
    \end{equation}
    
Writing the product in the form adapted to the $\sfs$-projection, 
$$\Sigma^B\cdot \Sigma^A=:\Sigma^{{A}'}=\big({A}'(v), v\big)\ ,$$
 we find  
    \begin{equation*}
        {A}'(v)=B(v')\,A( v')\ ,
    \end{equation*}
    where the function $v'=v'(v)$ is determined by the equation \smash{$v=\mathsf{Ad}^\ast_{B(v')}(v')$}.
    In the case that the bisection \smash{$\Sigma^B$} is close to the submanifold of units $e\times \mathfrak{g}^\ast$, or equivalently the image $\mathrm{im}(B)$ is contained in a small neighbourhood of the identity $e\in G$, one can solve this last equation approximately as $v'\approx \mathsf{Ad}^\ast_{B( v)^{-1}}(v)$. This gives an approximate expression for the potential 
    \begin{equation}\label{gth}
        A'(v)\approx B(v)\,A\big(\mathsf{Ad}^\ast_{A(v)^{-1}}(v)\big)\ .
    \end{equation}
    For a Lagrangian bisection $\Sigma^B$, this formula yields an infinitesimal gauge transformation of the gauge potential $A(v)$. 
    
The minimal coupling to the electromagnetic field is now described by the action functional
\begin{equation*}
    S=\int_\gamma\, \langle u, g^{-1}\,\dd g\rangle-\lambda\, H\big(g\,A(u), \mathsf{Ad}^\ast_{A(u)^{-1}}(u)\big)\ ,
    \end{equation*}
where $H(g,u)\in C^{\infty}(T^\ast G)$ is the Hamiltonian of a `free' particle. By construction, the action functional is invariant under the infinitesimal gauge transformations of the gauge potential (\ref{gth}) followed by the infinitesimal gauge transformations of the phase space variables 
\begin{equation*}
    (g',u')=(g,u)\,\big(\Sigma^B\big)^{-1}=\big(g\,B(u)^{-1}, \mathsf{Ad}^\ast_{B(u)}(u)\big)\ .
\end{equation*}
Here we used (\ref{guh}) along with $\big(\Sigma^B\big)^{-1}\approx \big(B(v)^{-1},v\big)$. 

\paragraph{$\boldsymbol{SU(2)}$ as momentum space.} 

Let us exemplify this construction  by the linear Poisson structure associated with the simplest non-abelian Lie algebra $\mathfrak{su}(2)$. The corresponding Poisson brackets read
\begin{equation*}
    \{x,y\}=z\quad,\quad  \{y,z\}=x\quad,\quad  \{z,x\}=y\ ,
    \end{equation*}
where $(x,y,z)$ is a vector of $\mathfrak{su}(2)^\ast\simeq  \mathbb{R}^3$. To make all constructions as explicit and simple as possible,  we will use the language of quaternions, see \cite[Ex.~14.22]{CFM}. Let $\mathbf{i}$, $\mathbf{j}$ and $\mathbf{k}$ be the imaginary unit quaternions satisfying
\begin{align*}
    \mathbf{i}^2=\mathbf{j}^2=\mathbf{k}^2=\mathbf{i}\,\mathbf{j}\,\mathbf{k}=-1 \ .
\end{align*}
We identify the group $SU(2)$ with the quaternions of norm one:
\begin{equation}\label{uvst}
    p=u+v\,\mathbf{i}+s\,\mathbf{j}+t\,\mathbf{k} \qquad \text{with} \quad u^2+v^2+s^2+t^2=1\ .
\end{equation}
These form the three-sphere $S^3$. The space $\mathfrak{su}(2)^\ast$ is then  identified with the purely imaginary quaternions via
\begin{equation}\label{q}
    q=x\,\mathbf{i}+y\,\mathbf{j}+z\,\mathbf{k}\ .
    \end{equation}
    
We can now define the Lie groupoid  $S^3\times \mathbb{R}^3\rightrightarrows \mathbb{R}^3$ with source and target maps (cf. (\ref{lpb}))
\begin{equation*}
    \sfs(p,q)=q\qquad\text{and}\qquad \sft(p,q)=p\, q\, \bar p\ ,
\end{equation*}
where $\bar p=u-v\,\mathbf{i}-s\,\mathbf{j}-t\,\mathbf{k}$ is the quaternion conjugate to (\ref{uvst}); by definition $p\,\bar p=1$. The product of two composable elements is given by 
\begin{equation*}
    (p,q)\,(p', p\,q\,\bar p)=(p\,p', q)\ .
\end{equation*}
The multiplicative symplectic form $\omega=\dd\theta$ is determined by the Liouville one-form 
\begin{align}\label{spt}
\begin{split}
\theta=-\tfrac12\,\mathrm{Tr}(q\, \bar p\, \dd p)&=-(x\,v+y\,s+z\,t)\,\dd u+(x\,u-y\,t+z\,s)\,\dd v \\
& \qquad +(x\,t+y\,u-z\,v)\,\dd s+(-x\,s+y\,v+z\,u)\,\dd t\ ,
\end{split}
\end{align}
where $\mathrm{Tr}(p)=p+\bar p=2\,u$ is the trace on the quaternion algebra.  

Each bisection  $\Sigma^A=\big(A(q), q\big)$ is determined  by a smooth map $A: \mathbb{R}^3\rightarrow S^3$ that assigns to every purely imaginary quaternion (\ref{q}) a quaternion 
\begin{equation}\label{AAAA}
A=A_0+A_1\,\mathbf{i}+A_2\,\mathbf{j}+A_3\,\mathbf{k}
\end{equation}
of norm one, $A\,\bar A=1$. The right translation by $\Sigma^A$ is given by
\begin{equation*}
   (p,q)\,\Sigma^A= (p,q)\,\big(A(p\,q\,\bar p), p\,q\,\bar p\big)=\big(p\,A(p\,q\,\bar p), q\big)\ .
\end{equation*}
Physically this formula defines gauge invariant coordinates on the phase space $S^3 \times \mathbb{R}^3$. The formula (\ref{fh}) for the composition of two bisections takes the form
\begin{equation}\label{comp1}
    \Sigma^B\cdot\Sigma^A=\big(B(q),  q \big)\,\big(A(B\,q\,\bar B),B\, q\,\bar B\big)=\big(B(q) \, A(B\,q\,\bar B), q\big)\ .
    \end{equation}
    
Consider a bisection $\Sigma^B$ that is infinitesimally close to the base manifold $\mathbb{R}^3$. This is defined by a map $B(q)=1+b(q)$, where  $b(q)$ is a small purely imaginary  quaternion. The bisection $\Sigma^B$ is Lagrangian if 
\begin{equation*}
    F^\sfs=-\tfrac12\,\mathrm{Tr}(\dd q\wedge \dd b)=0\ .
\end{equation*}
This last condition implies that 
\begin{equation*}
    b=\partial_x\varepsilon\, \mathbf{i}+\partial_y\varepsilon\, \mathbf{j}+\partial_z\varepsilon\, \mathbf{k}
    \end{equation*}
for some small scalar function $\varepsilon(x,y,z)$. 
On substituting $B=1+b$ in (\ref{comp1}), we get the expression for infinitesimal left translations on the group of bisections given by
\begin{equation}\label{gtA}
\delta_bA=b\,A+A(q+b\,q-q\,b)-A(q) \ .
\end{equation}

Let us treat the symplectic manifold  $(S^3\times\mathbb{R}^3, \omega )$ 
as the phase space  of a {\it non-relativistic} scalar particle moving in $\mathbb{R}^3$. Then the bisections (\ref{AAAA}) are naturally interpreted as the $3$-vector potentials for a purely magnetic field. The potential $A$ enjoys the infinitesimal gauge symmetry transformations (\ref{gtA}).
At this point, it is convenient to change variables $(p,q)\rightarrow (p, r=p\, q\, \bar p)$. In terms of these new variables, the symplectic potential (\ref{spt}) assumes the form $\theta =-\frac12\,\mathrm{Tr}(r\,\dd p\,\bar p)$ and the dynamics of a scalar particle minimally coupled to the external magnetic field is governed by the gauge invariant Lagrangian
\begin{equation}\label{Lsu2}
    L=-\tfrac12\,\mathrm{Tr}(r\,\dot p\, \bar p)-\lambda\,H\big(p\,A(r),\bar p\, r\,p\big)\ .
\end{equation}
Here $H(p,r)$ is the Hamiltonian of the free particle and the overdot stands for the time derivative. One should also remember that the `position' and `momentum' quaternions in (\ref{Lsu2}) are subject to the constraints  $r+\bar r=0$ and $p\,\bar p=1$; hence here the momentum space of a non-relativistic particle moving in $\mathbb{R}^3$
is compact and diffeomorphic to $S^3$. 

\paragraph{Poisson brackets of $\boldsymbol\kappa$-Minkowski type.} 

For another concrete example where everything can be computed explicitly, consider the linear Poisson structure on $\mathbb{R}^n$ of $\kappa$-Minkowski type, whose integrating symplectic groupoid was detailed in~\cite{Saemann:2012ab}. This is defined by the Poisson brackets
\begin{equation*}
    \{x^0,x^i\}=\kappa^{-1}\, x^i\qquad\text{and}\qquad \{x^i,x^j\}=0\ ,
\end{equation*}
where $i\in\{1,\dots,n-1\}$ and $\kappa$ is a constant. The integrating Lie group (momentum space) $G$ in this case is non-compact; for $n=2$ it is isomorphic to the connected affine group of $\mathbb{R}$, also known as the ``$a\,x+b$ group''.

The source map is given by
\begin{equation*}
    \sfs^0(p,q)=q^0-\kappa^{-1}\,q^i\,p_i=x^0\qquad\text{and}\qquad  \sfs^i(p,q)= q^i=x^i\ ,
\end{equation*}
while the target map is
\begin{equation*}
    \sft^0(p,q)=q^0\qquad\text{and}\qquad  \sft^i(p,q)=\ee^{-\kappa^{-1}\,p_0}\, q^i\ ,
\end{equation*}
where $\{q^\mu,p_\nu\}=\delta^\mu_\nu$ with $\mu, \nu \in\{0,1,\dots,n-1\}$.
Two elements $(p,q)$ and $(p^\prime,q^\prime)$ are composable if $\sft(p,q)=\sfs(p^\prime,q^\prime)$, or equivalently
\begin{equation*}
    (q^\prime)^0=q^0+\kappa^{-1}\,\ee^{-\kappa^{-1}\,p_0}\,q^i\,p^\prime_i\qquad\text{and}\qquad  (q^\prime)^i=\ee^{-\kappa^{-1}\,p_0}\, q^i\ .
\end{equation*}
The product reads
\begin{align*}
&\big((p_0, p_i)\,,\,(q^0, q^i)\big)\,\big((p^\prime_0, p_i^\prime)\,,\,(q^0+\kappa^{-1}\,\ee^{-\kappa^{-1}\,p_0}\,q^i\,p^\prime_i,\ee^{-\kappa^{-1}\,p_0}\, q^i)\big)\\[4pt]
   & \hspace{4cm} =\big((p_0+p^\prime_0,  p_i+\ee^{-\kappa^{-1}\,p_0}\, p^\prime_i)\,,\,(q^0+\kappa^{-1}\,\ee^{-\kappa^{-1}\,p_0}\,q^i\,p^\prime_i, q^i)\big) \ .
\end{align*}

Let us write the gauge invariant phase space coordinates. For this, we consider a bisection $\Sigma^A=\big(A(\widecheck q),\widecheck q\big)$ and apply the right translation
\begin{equation*}
    (p,q)\,\big(A(\widecheck q),\widecheck q\big)=:(P,Q)\ .
\end{equation*}
Composability implies 
\begin{eqnarray*}
    (\widecheck q)^0=q^0+\kappa^{-1}\,\ee^{-\kappa^{-1}\,p_0}\,q^i\,A_i(\widecheck q)\qquad\text{and}\qquad (\widecheck q)^i=\ee^{-\kappa^{-1}\,p_0}\,q^i \ .
\end{eqnarray*}
Then using the product one finds
\begin{eqnarray*}
Q^0=(\widecheck q)^0\quad,\quad Q^i=q^i\qquad\text{and}\qquad P_0=p_0+A_0(\widecheck q)\quad,\quad P_i=p_i+\ee^{-\kappa^{-1}\,p_0}\,A_i(\widecheck q)\ .
\end{eqnarray*}
The minimal coupling of the particle to the electromagnetic field is thus described by the Lagrangian
\begin{align*}
    L = p_i\,\dot q^i - \lambda\,H(P,Q) \ .
\end{align*}

\subsection{A class of quadratic Poisson brackets}

As our final example, we consider the log-canonical Poisson brackets on $\mathbb{R}^n$, see \cite{li2020symplectic} and \cite[Ex.~14.23]{CFM}, which are given by 
\begin{equation*}
    \{x^i,x^j\}=\alpha^{ij}\,x^i\,x^j \qquad \mbox{(no sum over $i,j$)}\, .
\end{equation*}
Here $\alpha^{ij}=-\alpha^{ji}$ are constants. The integrating symplectic groupoid is given by the space $\mathbb{R}^{2n}=\mathbb{R}^n\times \mathbb{R}^n$ with coordinates $(x^1,\ldots, x^n,p^1,\ldots,p^n)$ and the symplectic structure
\begin{equation}\label{xx}
    \omega=\sum_{i=1}^n\, \dd x^i\wedge \dd p^i-\frac12\,\sum_{i,j=1}^n\,\alpha^{ij}\,\dd(x^i\,p^i)\wedge \dd(x^j\, p^j)\ .
\end{equation}
The corresponding Poisson brackets read
\begin{equation*}
    \{x^i,x^j\}=\alpha^{ij}\,x^i\,x^j\quad,\quad \{p^i,p^j\}=\alpha^{ij}\, p^i\,p^j\quad,\quad \{x^i, p^j\}=\delta^{ij} -\alpha^{ij}\,x^i\,p^j\ .
\end{equation*}
Unlike the symplectic realization (\ref{se}), the symplectic structure (\ref{xx}) yields non-zero Poisson brackets among the momentum coordinates $p$.

The source and target maps are defined by 
\begin{equation*}
    \sfs(x^i,p^i)=(x^i)\qquad\text{and}\qquad \sft(x^i,p^i)=\Big(x^i\,\ee^{\,\sum\limits_{j=1}^n \, \alpha^{ij}\,x^j\,p^j}\Big)\ .
\end{equation*}
The product of composable elements has the form
\begin{equation}\label{xp}
    (x^i,p^i)\,( \widecheck x{}^i,\widecheck p{}^i)=\Big(x^i, p^i+\widecheck p{}^i\,\ee^{\,\sum\limits_{j=1}^n\alpha^{ij}\,x^j\,p^j}\Big)\ .
\end{equation}
Composability implies
\begin{align}\label{barx}
\widecheck x{}^i=x^i\,\ee^{\,\sum\limits_{j=1}^n\, \alpha^{ij}\,x^j\,p^j}\ .
\end{align}
It is straightforward to check the brackets 
\begin{equation}\label{xxp}
    \{\widecheck x{}^i, x^j\}=0\qquad\text{and}\qquad \{\widecheck x{}^i, p^j\}=\delta^{ij}\,\ee^{\,\sum\limits_{k=1}^n\,\alpha^{ik}\,x^k\,p^k}\ .
\end{equation}
In particular, the map $\sfs$ is Poisson by definition, while $\sft$ is anti-Poisson:
\begin{equation*}
 \{\sfs^i,\sfs^j\}=\alpha^{ij}\,\sfs^i\,\sfs^j \qquad\text{and}\qquad  \{\sft^i,\sft^j\}=-\alpha^{ij}\,\sft^i\,\sft^j\ .
\end{equation*}

Let  $\Sigma^A$ be the bisection defined by a graph $\big(\widecheck x, A(\widecheck x)\big)\in \mathbb{R}^n\times \mathbb{R}^n$. Then the right translation by $\Sigma^A$ gives
\begin{equation*}
    (x^i,p^i)\,\Sigma^A=(x^i, P^i)\ ,
\end{equation*}
where 
\begin{equation}\label{gip}
P^i=p^i+A^i(\widecheck x)\, \ee^{\,\sum\limits_{j=1}^n\, \alpha^{ij}\,x^j\,p^j} \ ,
\end{equation}
and the argument $\widecheck x$ is determined from the composability condition as in \eqref{barx}.
Notice that the right translations preserve each $\sfs$-fiber, that is, points move along $\sfs$-fibers under the right action of bisections. 
As a result, the gauge invariant positions coincide with the original coordinates $x$. Only the momenta are modified according to (\ref{gip}).

The covariant field strength tensor $F^\sfs=\frac12\, F_{lk}^\sfs\, \dd x^l\wedge \dd x^k$ has components
\begin{equation}\label{Fxx}
    F^\sfs_{lk}=\partial_k A^l-\partial_l A^k -\sum_{i,j=1}^n\, \alpha^{ij}\,\big(A^i\,\delta^i_l+x^i\,\partial_lA^i\big)\,\big(A^j\,\delta^j_k+x^j\,\partial_kA^j\big)\ .
\end{equation}
Let $\Sigma^B=\big(x, B(x)\big)$ for some infinitesimal functions $B^i(x)$. It follows from (\ref{Fxx}) that $\Sigma^B$ is a Lagrangian bisection if $\partial_i B^j-\partial_j B^i=0$. 
Using (\ref{xp}) one can compute 
\begin{equation*}
    \Sigma^B\cdot \Sigma^A = \Big(x^i, B^i(x) +A^i(\widecheck x_B)\, \ee^{\,\sum\limits_{j=1}^n\,\alpha^{ij}\,x^j\,B^j(x)}\Big )\ ,
\end{equation*}
where
\begin{equation*}
  \widecheck x{}^i_B=x^i\, \ee^{\,\sum\limits_{j=1}^n\, \alpha^{ij}\,x^j\,B^j(x)}  \ .
\end{equation*}
Taking $B^i(x)=\partial_i\varepsilon (x)$, we obtain the infinitesimal gauge transformation 
\begin{equation}\label{Axx}
    \delta_{\varepsilon} A^i=\sum_{j=1}^n\,\big (\delta^{ij}- x^j\,\alpha^{ji}\, A^{i}(x)\big)\,\partial_j\varepsilon  +\{A^i,\varepsilon\}\ .
\end{equation}

The minimal coupling of the particle to the external electromagnetic field is now described by the Lagrangian 
\begin{equation}\label{Lxx}
    L=\sum_{i=1}^n\, p^i\,\dot x^i+\sum_{i,j=1}^n\, x^i\,p^i\,\alpha^{ij}\,\big(\dot x^j\,p^j+x^j\,\dot p^j\big)-\lambda\, H(x,P)\ ,
\end{equation}
where the gauge invariant momenta $P^i$ are defined by (\ref{gip}) and (\ref{barx}).  By construction, this Lagrangian is invariant (up to total derivative) under the combined gauge transformations of the electromagnetic field (\ref{Axx}) and the canonical transformations of the phase space variables given by
\begin{equation*}
    \delta_{\varepsilon} x^i=\{x^i,\varepsilon (\widecheck x)\}=0\qquad\text{and}\qquad 
   \delta_\varepsilon p^i=\{p^i,\varepsilon(\widecheck x)\}= -\partial_i\varepsilon(\widecheck x)\,\ee^{\,\sum\limits_{k=1}^n\,\alpha^{ik}\,p^k\,x^k}\ .
\end{equation*}
Here we used (\ref{xxp}). Of course, by changing the dynamical variables as $(x,p)\rightarrow (\widecheck x, p)$ one can bring 
the Lagrangian (\ref{Lxx}) into a form where the electromagnetic potentials $A(\widecheck x)$ depend on the `position coordinates' $\widecheck x$. Contrary to $x$, however, the new position coordinates $\widecheck x$ are not gauge invariant when $\alpha^{ij}\neq 0$. 

\section{Conclusions}
\label{sec:concl}

\paragraph{Summary of results.} 

In this paper we have interpreted the classical phase space of a point particle on a noncommutative spacetime, regarded as the quantization of a Poisson manifold, as the corresponding integrating symplectic groupoid. This provides a conceptual geometric explanation for the structure of the gauge sector of Poisson electrodynamics through the role of bisections of symplectic groupoids. We used this description to explicitly realise the gauge invariant minimal coupling of a dynamical charged particle to a background gauge field in Poisson electrodynamics. We constructed both gauge covariant as well as gauge invariant field strength tensors which are related to one another through the actions of bisections, and briefly addressed the problem of formulating gauge invariant action functionals for the electromagnetic field. 

The geometric realisations of these physical quantities in terms of symplectic groupoids are concisely summarised in Table~\ref{tab:glossary}. From a conceptual point of view, our constructions provide a precise mathematical meaning to earlier proposals of the role of curved momentum spaces in noncommutative geometry. As one particular phenomenon anticipated from early studies of noncommutative gauge theories, we explicitly demonstrated the interplay between gauge transformations and spacetime diffeomorphisms, as well as the related appearance of emergent gravity phenomena in Poisson electrodynamics.

\paragraph{Further extensions and applications.} 

A tantalizing perspective of our approach stems from the fact that a Poisson manifold can be quantized by a twisted polarized convolution $C^*$-algebra of a symplectic groupoid, extending standard geometric quantization of symplectic manifolds; see e.g.~\cite{Hawkins2008,Saemann:2012ab} and references therein. It would be interesting to exploit this quantization to produce a \emph{complete} noncommutative gauge theory for the first time, beyond the semi-classical sector captured by Poisson electrodynamics. Among other things, this approach has the prospect of elucidating the geometric origin and algebraic structure of full noncommutative gauge symmetries.

Another interesting avenue would be to generalise our geometric perspective to the case of \emph{twisted} Poisson manifolds, which are integrated by a twisted version of symplectic groupoids~\cite{Cattaneo:2003fs}. Twisted Poisson brackets are quantized by \emph{nonassociative} star-products and arise in the dynamics of electric charges in smooth magnetic monopole backgrounds as well as in string theory on certain flux compactifications; see~\cite{Szabo:2019hhg} for a review and references therein. The symplectic embeddings in these instances are described in~\cite{Kupriyanov:2018xji,Kupriyanov:2018yaj,Kupriyanov_2021} and the gauge sector of twisted Poisson electrodynamics in~\cite{Kupriyanov_2021}. It would be interesting to understand how these constructions are explained geometrically in terms of the twisted symplectic groupoids of~\cite{Cattaneo:2003fs}.

Our proposed construction of minimal interaction allows, {\it in principle},  experimental verification of the hypothesis of spacetime noncommutativity. 
For example, one can test the deviation from Coulomb's law for a pair of charged particles due to one or another form of noncommutativity. 
Upon quantization such a deviation should inevitably manifest itself in atomic spectra; even if the corresponding measurements are beyond the reach of current experiments. Another option is to study the magnetic susceptibility of an electron gas within the framework of our proposed model of minimal interaction. As is well known, a non-zero magnetic susceptibility is a purely quantum effect due to the noncommutativity of position coordinates and momenta (Landau diamagnetism).   The noncommutativity of position coordinates themselves should contribute to the overall susceptibility of the electron gas. This provides another possible test for noncommutativity.

Our results may also have implications for the hypothesis of Lorentz-invariance violation, which has been a topic of intensive investigation in both theoretical and experimental high energy physics, see e.g.~\cite{colladay1998lorentz,coleman1999high,kostelecky2004gravity, liberati2013tests,abreu2022testing,Sarker2023InvestigatingTE, Finke_2023}. At low energies, one can describe hypothetical Lorentz violating effects  by an effective field theory called a `standard model extension'. The idea is to include into the standard model Lagrangian all possible terms that, while violating Lorentz and CPT invariance, do respect the standard gauge invariance and renormalization properties; in so doing, various tensor coefficients responsible for breaking Lorentz invariance are interpreted as expectation values in some fundamental (yet unknown) theory. The case of the pure electromagnetic field was considered in~\cite{kostelecky2009electrodynamics}. Poisson electrodynamics offers a refined interpretation of the Lorentz-invariance violating terms -- all of them arise from a background Poisson bivector. Furthermore, unlike the various standard model extensions, the violation of the Lorentz symmetry is accompanied here with a non-trivial deformation of the gauge group. All this makes Poisson electrodynamics quite rigid and therefore more specific with respect to the predictions of Lorentz violation.

Poisson electrodynamics features a remarkable duality between the noncommutativity of spacetime and the curvature of momentum space.
As discussed in \S \ref{LPBr}, the simplest $\mathfrak{su}(2)$-type Poisson brackets leads to a model of a non-relativistic particle whose momenta form a three-sphere $S^3$, rather than $\mathbb{R}^3$. The radius of the momentum sphere turns out to be the reciprocal to the uncertainty in measurements of spacetime coordinates. The compactness of the momentum space has numerous far-reaching physical consequences: from the existence of a natural ultraviolet cut-off in quantum field theory and negative-temperature states in classical thermodynamics to deviations in Planck's radiation law and caloric properties of gases at very high temperatures \cite{Born}. 
More recent proposals for experimental tests of the curvature of momentum space can be found in \cite{Amelino}. 
It should be noted that all previous works on curved momentum space treated the curvature exclusively in terms of Riemannian or affine geometry. Our analysis clearly demonstrates that the true geometry behind curved momentum spaces is that of symplectic groupoids.

\section*{Acknowledgements} 

We would like to acknowledge an enlightening correspondence with Alan Weinstein. This article is based upon work from COST Action CaLISTA CA21109 supported by COST (European Cooperation in Science and Technology). R.J.S. thanks the Centro de Matem\'atica, Computa\c{c}\~ao e Cogni\c{c}\~ao of the Universidade de Federal do ABC for hospitality and support during the initial stages of this work. V.G.K. acknowledges support from the CNPq Grant 304130/2021-4. The work of V.G.K. and {R.J.S.} was supported in part by the FAPESP Grant 2021/09313-8. The work of A.A.S. was partially supported by the FAPESP Grant 2022/13596-8 and by the Foundation for the Advancement of Theoretical Physics and Mathematics ``BASIS''. The results of \S 5.3 were obtained under the exclusive support of the Ministry of Science and Higher Education of the Russian Federation (project No. FSWM-2020-0033).

\end{document}